\documentclass[11pt]{article}

\usepackage[utf8]{inputenc}
\usepackage[T1]{fontenc}
\usepackage{amsmath,amssymb,amsfonts,amsthm}
\usepackage{graphicx}
\usepackage{booktabs}
\usepackage[dvipsnames]{xcolor}
\usepackage[margin=1in]{geometry}
\usepackage[authoryear,round]{natbib}
\usepackage{hyperref}
\definecolor{darkblue}{RGB}{0,51,102}
\hypersetup{colorlinks=true,citecolor=darkblue,linkcolor=darkblue,urlcolor=darkblue}

\newtheorem{proposition}{Proposition}
\newtheorem{remark}{Remark}
\newtheorem{example}{Example}

\theoremstyle{plain}

\newcommand{\bpi}{\boldsymbol{\pi}}
\newcommand{\btheta}{\boldsymbol{\theta}}
\newcommand{\bw}{\boldsymbol{w}}
\newcommand{\bx}{\boldsymbol{x}}
\newcommand{\by}{\boldsymbol{y}}
\newcommand{\bz}{\boldsymbol{z}}
\newcommand{\bX}{\boldsymbol{X}}

\title{Quantifying Portfolio Demutualization:\\A Benchmark-Relative Pooling--Profiling Scale}
\author{%
Arthur Charpentier\thanks{Universit\'e du Qu\'ebec \`a Montr\'eal, Canada; Kyoto University, Japan; Chaire PARI (Institut Europlace de Finance--ILB, ENSAE Paris \& Sciences Po), France.}
\and
Laurence Barry\thanks{Chaire PARI (Institut Europlace de Finance--ILB, ENSAE--CREST Paris \& Sciences Po), France.}%
}
\date{}

\begin{document}
\maketitle

\begin{abstract}
Insurance pricing combines pooling with differentiation: a tariff may leave benchmark differences in expected loss partly mutualized or translate them into policy-level premium differences. We propose a benchmark-relative pooling--profiling scale with two complementary coordinates. The coupled $L^p$ coordinate measures policy-level alignment between an evaluated tariff and a stated benchmark pure premium, whereas the marginal Wasserstein coordinate compares their exposure-weighted premium distributions. The difference between their residual $p$-costs defines an allocation mismatch. Under portfolio balance, the coupled $L^1$ coordinate has an exact actuarial interpretation: it is the fraction of the transfer volume induced by full pooling that the tariff removes. Synthetic and motor-insurance applications show that broad classes, proxies, shrinkage and tail caps can affect marginal differentiation, policy-level allocation and transfers differently. A barycentric group-parity intervention further shows that conditional premium disparities can fall mainly through reallocation and restored benchmark-relative transfers, with little change in marginal differentiation. The framework is descriptive, fixed-portfolio and explicitly conditional on the chosen benchmark.
\end{abstract}


\section{Introduction}
\label{sec:introduction}

Insurance has always been a classificatory institution. It transforms uncertain individual losses into collectively manageable risks by gathering policyholders into pools, estimating expected losses within those pools, and charging premiums that make the transfer of risk financially viable. Pooling is not merely a technical convenience. It is constitutive of insurance: without a pool, the law of large numbers has no object; without some residual uncertainty about who will suffer a loss, pooling is impossible. The exchange ceases to be insurance in the ordinary sense and becomes saving, credit, or direct loss financing. Yet insurance pricing has also always relied on segmentation. Insurers partition policyholders into classes regarded as sufficiently homogeneous for pricing purposes and adjust premiums according to observed or inferred differences in expected loss.

This dual nature creates permanent tension. Pooling is the mechanism through which uncertainty is mutualized and shared, whereas segmentation is the mechanism through which cross-subsidies between heterogeneous risks are reduced. Insurance pricing therefore lies on a continuum between two limiting poles. At one pole, a fully pooled tariff charges everyone the same premium and leaves all observable risk heterogeneity unpriced. At the other, a fully profiled tariff charges each policyholder according to a stated estimate of their individual risk. Most actual tariffs lie between these extremes. They neither erase risk differences altogether nor translate every available signal into price. Instead, they create pools of different sizes and leave different amounts of heterogeneity unpriced within them. The central aim of this article is to make this intermediate position measurable through a coherent family of indices.

We call this continuum the pooling--profiling scale. The phrase captures a shift that has been described in sociological and actuarial discussions of insurance but rarely formalized as a set of measurable portfolio properties. The contemporary promise of big data and predictive analytics is that insurers may move from risk classes to risk profiles, from group-level pricing to much more granular risk assessment, and from mutualized uncertainty to behavior-based or algorithmically differentiated premiums \citep{BarryCharpentier2020,CevoliniEsposito2020}. Insurers often present this move as progress toward greater accuracy and actuarial fairness \citep{Meyers2018,meyersvanhoyweghen2018}, as each insured would pay for their ``own'' risk rather than for the risk of others. But this formulation hides the distributive question. A price closer to individualized predicted loss is not simply more accurate; it also changes the pattern of transfers within the portfolio, the size and meaning of risk pools, and potentially the accessibility of insurance for high-risk or vulnerable groups.

Indices of mutualization and demutualization are therefore timely. Behind the promise of individualization through behavioral insurance \citep{JeanningrosMcFall2020}, insurers' pricing practices remain largely opaque for reasons of competition and trade secrecy. Yet a pricing model governs who contributes to the pool, who benefits from it, who may be priced out of it, and which differences are recognized as legitimate. Measuring demutualization is a necessary step toward making the trade-off between actuarial adequacy and solidarity observable.

This paper develops a benchmark-relative framework for quantifying portfolio demutualization. We use the term in a pricing sense: demutualization is the movement from a pooled premium toward a stated profiling benchmark, together with the removal of benchmark-relative transfers. The paper makes three contributions. First, it defines explicit pooling and profiling endpoints and calibrates a canonical path between them. Second, it introduces two complementary coordinates. A coupled coordinate compares each policyholder's premium with their own benchmark risk, whereas a distributional Wasserstein coordinate compares the marginal monetary distributions of premiums and benchmark risks. Their divergence is an allocation mismatch, separating the amount of differentiation from the identities of those who bear it. Third, under portfolio balance, the coupled coordinate has an exact transfer interpretation: it equals the fraction of the transfer volume generated by full pooling that the tariff removes. Its complement is therefore a direct measure of remaining benchmark-relative mutualization.

The distinction between the two coordinates is central. A tariff may reproduce the benchmark distribution while allocating it to the wrong policyholders; distributional similarity alone is not enough. Another tariff may remain strongly coupled to the benchmark while compressing all differences toward the mean; in that case, individual alignment is preserved but substantial mutualization remains. The pooling--profiling scale is therefore not a unique scalar ordering of tariffs. It is a common normalization on which coupled and distributional measures provide complementary coordinates. Group transfers, Gini and calibration diagnostics, benchmark sensitivity, and uncertainty quantification then help interpret the monetary mechanisms behind those coordinates.

The empirical application shows why this distinction matters. Starting from a fixed motor insurance portfolio and an operational profiling benchmark, we compare pooled, GLM, boosted tree, shrinkage and benchmark tariffs. The proposed measures quantify how much benchmark heterogeneity is monetized, how closely the resulting premiums are allocated to the same policies, and how much of the full-pooling transfer volume remains. A barycentric group parity correction provides a final illustration: it sharply reduces gender-conditional premium differences mainly through premium reallocation and the restoration of benchmark-relative transfers, rather than through a broad compression of the marginal premium distribution.

The proposed quantities are diagnostic, not normative. They do not identify an optimal degree of mutualization, decide which rating factors are legitimate, or incorporate endogenous participation and coverage choices. They make visible how a stated pricing rule converts an explicit representation of risk heterogeneity into prices and transfers on a fixed portfolio. Technical adequacy, solidarity, non-discrimination and affordability can then be assessed as distinct objectives rather than compressed into a single tariff statistic.

Section~2 describes the motivation and the relevance in the context of big data. Section~3 introduces the benchmark-relative coordinates and their main properties. Section~4 develops the operational diagnostics, inference and the group-parity intervention. Sections~5 and~6 present the synthetic and motor-insurance applications.

\section{Insurance pricing between pooling and profiling}
\label{sec:pooling-profiling}

Modern insurance emerged with the statistical treatment of uncertainty \citep{BarryCharpentier2020}. Its basic operation consists in transforming events that are unpredictable at the individual level into statistical regularities that can be estimated at the collective level. In this sense, insurance depends on a specifically statistical imagination of the future: uncertainty is not eliminated but becomes manageable when adopting the aggregate viewpoint of the population. The individual case is not ignored but treated as one observation within a larger regularity. Insurance pricing therefore did not begin with an individual prediction, but with the construction of classes within which individuals could be considered equivalent for risk calculation.

Insurance offers protection precisely because no one knows in advance who among the policyholders will be unlucky. The premiums paid by those who do not suffer losses finance the claims of those who do. This ex-post redistribution is not a defect of insurance; it is its core mechanism. \citet{ThieryVanSchoubroeck2006} refer to this as ``chance solidarity'': policyholders with comparable risk profiles share the aleatoric occurrence of the event with the other members of their class \citep[see also][]{LehtonenLiukko2015}.

This helps explain why the language of pooling can be misleading when used too broadly. Pooling does not necessarily mean charging everyone the same premium. A highly segmented portfolio remains pooled within each tariff class, thereby preserving chance solidarity inside the class. Conversely, a uniform price across heterogeneous risks creates risk solidarity \citep{LehtonenLiukko2015} but may also generate adverse selection. Indeed, if individuals are pooled without regard to observable and material differences in expected loss, lower-risk policyholders might perceive themselves as unduly charged. In competitive markets, they may become targets for competitors and leave the pool, weakening the balance of the remaining portfolio. Insurance therefore depends on a delicate construction: policyholders must be similar enough for a common price to be technically plausible, but not so perfectly known that the uncertainty supporting the pool disappears.

Classification is the standard actuarial response to heterogeneity and the risk of adverse selection. The purpose is to form classes within which policyholders are sufficiently similar and between which expected loss differences justify different premiums. Segmentation is therefore not external to mutualization. It is the way modern insurance makes mutualization technically possible under heterogeneity.

Yet classification does not simply reveal pre-existing homogeneous groups. It also constructs them. As \citet{BarryCharpentier2020} argue, insurance classes are produced by statistical practices: by the information collected, the categories through which individuals are described, and the modelling conventions or regulations through which some differences become pricing-relevant while others remain pooled. Traditional segmentation is also constrained by practical and statistical limits. Each class must contain enough exposure for experience to be credible. A rating factor is therefore a device that makes some dimensions of difference visible and actionable. The resulting class is both an empirical object and a pricing convention. This constructed character matters for demutualization. A more granular tariff does not necessarily approach a unique or intrinsic notion of individual risk. It changes how heterogeneity is represented and allocated, creates new boundaries between policyholders, and therefore creates new patterns of pooling and transfer.

Bonus--malus systems illustrate the intermediate nature of actuarial segmentation. They introduce individual claims experience into pricing and thereby differentiate policyholders within broader classes, but they do not abolish the class logic. They refine it. Past claims provide a signal of unobserved heterogeneity, while the resulting relativities remain embedded in a collective tariff structure. Similarly, credibility theory does not eliminate pooling; it governs the balance between individual or class-specific experience and collective experience. The actuarial tradition has therefore long operated between two principles: differentiate when credible evidence of heterogeneity is available, and pool when individual information is insufficient, unstable, prohibited or judged irrelevant.

The pooling--profiling scale we propose here builds on this observation. It does not oppose pooling to segmentation. Rather, it treats segmentation as a movement away from complete pooling and toward profiling. A tariff with a few broad classes leaves substantial heterogeneity pooled within classes. A highly granular tariff reduces intra-class heterogeneity and risk solidarity. The balance between these two poles has social consequences because segmentation can intensify exclusions and affordability pressures, but it can also create incentives for prevention. A very large pooling means broad risk solidarity, but when high-risks are also the most well-off, it might end up with lower-income policyholders subsidizing higher-income policyholders. \citet{FrezalBarry2019} argue that actuarial fairness is often misunderstood when it is treated as the uniquely fair answer to these issues. Expected loss alignment might be a professional and technical ideal, but it cannot by itself determine which risk differences should be priced \citep[see also][]{Arrow1963,Barry2023,meyersvanhoyweghen2018}.

The contemporary debate on demutualization has been intensified by big data, machine learning and connected devices. Telematics in motor insurance, wearable devices in health and life insurance, smart-home sensors in property insurance, credit-based insurance scores and largescale behavioral data all promise more accurate and more dynamic risk classification. The traditional tariff class may be replaced, or at least supplemented, by a score that is continuously updated and potentially specific to the individual.

The shift to profiling has important consequences. Indeed, ``behavior-based'' pricing can change the function of insurance. If premiums are continuously adjusted according to behavior, insurance may become less a protection against misfortune than a system of incentives operating through feedback loops \citep{TanninenLehtonenRuckenstein2021,TanninenLehtonenRuckenstein2022}. Besides, big data also changes the informational structure of insurance markets. In classical adverse selection models, policyholders know more about their own risk than insurers do. Data-intensive insurance can partly reverse this asymmetry: insurers may infer risk-relevant information that policyholders themselves do not know or cannot interpret \citep{BrunnermeierLambaSeguraRodriguez2025,ElingGemmoGuxhaSchmeiser2024}. This reversal may reduce adverse selection and moral hazard, but may also increase insurers' capacity to classify, target, exclude or discipline policyholders.

For our purposes, the key lesson is that profiling is not a single technological state but includes a myriad of situations. A portfolio may become more profiled because behavioral variables replace demographic ones, because premiums become more dispersed, because classes become smaller, because within-class heterogeneity falls, or because transfers between benchmark risk levels shrink. These phenomena are related but not identical. They must be measured from the monetary tariff rather than inferred from the number of variables or the complexity of the model.

Our pooling--profiling scale therefore comes to measure the demutualization effect associated with more granular data included into pricing. It does not decide where the optimal position lies. It makes visible how far a portfolio has moved, which dimensions of mutualization have been reduced, and which policyholders bear the consequences. In this article, we use pricing demutualization in a narrow sense. Relative to a stated benchmark, a portfolio becomes more demutualized when its pricing structure moves from pooling toward profiling: premiums track benchmark risk estimates more closely, less benchmark heterogeneity remains pooled, and the transfers generated by full pooling are reduced.

This approach has two advantages. First, it permits comparisons across products, firms, periods and regulatory scenarios. One can ask how adding telematics, removing gender, capping a rating factor or changing territorial relativities moves the same portfolio along the pooling--profiling scale. Second, it connects conceptual discussions of personalization to actuarial practice. Rather than asking whether big data ends insurance, it asks how much benchmark-relative mutualization remains under a stated tariff.

The object studied below is fixed-portfolio and benchmark-relative. It compares alternative premium vectors for the same policyholders and exposure weights. It does not incorporate changes in participation, coverage or market composition, and it does not claim that the benchmark is a neutral representation of true individual risk. These restrictions define the object being measured: the monetary organization of mutualization under a stated pricing rule. The next section turns this object into coupled and distributional coordinates, together with an exact transfer interpretation for the coupled $L^1$ case.

\section{Measuring benchmark-relative demutualization}
\label{sec:coordinates}
\label{sec:setup}

\subsection{Portfolio, benchmark, and balance}
\label{subsec:portfolio-benchmark}

Consider $n$ policies with raw exposures $e_i>0$ and normalized exposure weights
\[
w_i=\frac{e_i}{\sum_{j=1}^n e_j},
\qquad
\sum_{i=1}^n w_i=1.
\]
Write $\bw=(w_1,\ldots,w_n)$, and let
\[
\btheta=(\theta_1,\ldots,\theta_n),
\qquad
\bpi=(\pi_1,\ldots,\pi_n)
\]
denote the benchmark and evaluated technical pure-premium vectors. In a simulation, $\theta_i$ may be an oracle conditional mean. In an empirical audit, $\btheta$ is an operational reference defined by a stated covariate set, model class, tuning rule and out-of-sample protocol. All conclusions below are conditional on that choice.

Define the exposure-weighted benchmark mean
\[
\bar\theta=\sum_{i=1}^n w_i\theta_i,
\]
and let $\boldsymbol 1$ denote the vector of $n$ ones. We assume throughout the main paper that the benchmark is non-degenerate: $\theta_i$ is not constant over the policies with positive exposure weight. If it is constant, the two endpoints below coincide and no pooling--profiling coordinate is defined.

A premium vector $\bpi$ is \emph{balanced relative to the benchmark} if
\begin{equation}
\label{eq:balance}
\sum_{i=1}^n w_i\pi_i=\bar\theta.
\end{equation}
Balance is not needed to define the distance-based coordinates. It is imposed in the actuarial audit to prevent changes in total premium volume from being confounded with changes in differentiation, and it is required for the exact transfer interpretation in Proposition~\ref{prop:transfer-coupled-l1}.

\subsection{Pooling and profiling endpoints}
\label{subsec:endpoints}

The \emph{pooling endpoint} charges the benchmark mean to every policy,
\[
\bpi^{\rm pool}=\bar\theta\boldsymbol 1,
\]
whereas the \emph{profiling endpoint} charges the benchmark policy by policy,
\[
\bpi^{\rm prof}=\btheta.
\]
These labels refer to the pricing of benchmark heterogeneity, not to the disappearance of insurance pooling itself. Even at the profiling endpoint, idiosyncratic claim uncertainty remains pooled ex post; what changes is whether differences in the stated expected-loss benchmark are left unpriced or translated into ex-ante premium differences.

For calibration, define the \emph{canonical pooling--profiling path}
\begin{equation}
\label{eq:canonical-path}
\bpi^{(a)}=(1-a)\bar\theta\boldsymbol 1+a\btheta,
\qquad 0\leq a\leq1.
\end{equation}
Along this path, $a$ is the share of benchmark heterogeneity translated linearly into premium differences. The construction resembles credibility-type shrinkage toward a collective mean, but the role of $a$ here is geometric rather than inferential \citep{WuthrichMerz2023}. The two coordinates below are normalized to return exactly $a$ on this path.

The scale uses two complementary views of the same tariff. The first keeps every evaluated premium paired with the benchmark premium of the same policy. The second ignores policy identities and compares only marginal monetary distributions. We define the two working coordinates directly for $p\geq1$; a more general construction based on positively homogeneous discrepancies is given in Appendix~\ref{sup:general-discrepancies}.

\subsection{Policy-level coupled coordinate}
\label{subsec:coupled}

For $\bz=(z_1,\ldots,z_n)$, define the exposure-weighted $L^p$ norm
\begin{equation}
\label{eq:weighted-lp}
\|\bz\|_{p,\bw}
=
\left(\sum_{i=1}^n w_i|z_i|^p\right)^{1/p}.
\end{equation}
We call the first coordinate \emph{coupled} because $\pi_i$ remains paired with $\theta_i$ for the same policy. Define
\begin{equation}
\label{eq:coupled-coordinate}
D_p^{\rm cpl}(\bpi;\btheta,\bw)
=
1-
\frac{\|\bpi-\btheta\|_{p,\bw}}
{\|\bar\theta\boldsymbol 1-\btheta\|_{p,\bw}}.
\end{equation}
The denominator is exactly the policy-level discrepancy produced by the pooling endpoint.

\begin{proposition}[Properties of the coupled profiling coordinate]
\label{prop:coupled-properties}
Let $p\geq1$ and let $\btheta$ be non-degenerate. Then
$D_p^{\rm cpl}$ satisfies:

\begin{itemize}
    \item[-] endpoint normalization,
    \[
    D_p^{\rm cpl}(\bar\theta\boldsymbol 1;\btheta,\bw)=0,
    \qquad
    D_p^{\rm cpl}(\btheta;\btheta,\bw)=1;
    \]

    \item[-] canonical-path calibration,
    \[
    D_p^{\rm cpl}\bigl(
    (1-a)\bar\theta\boldsymbol 1+a\btheta;
    \btheta,\bw
    \bigr)
    =a,
    \qquad 0\leq a\leq1;
    \]

    \item[-] invariance under a common positive monetary rescaling;

    \item[-] invariance under joint relabelling of policy values and
    exposure weights;

    \item[-] continuity away from benchmark degeneracy.
\end{itemize}

Moreover,
\begin{equation}
\label{eq:coupled-one-characterization}
D_p^{\rm cpl}(\bpi;\btheta,\bw)=1
\quad\Longleftrightarrow\quad
\pi_i=\theta_i
\quad\text{for every }i\text{ with }w_i>0.
\end{equation}
Thus the upper endpoint uniquely characterizes policy-level profiling.

By contrast,
\begin{equation}
\label{eq:coupled-zero-characterization}
D_p^{\rm cpl}(\bpi;\btheta,\bw)=0
\quad\Longleftrightarrow\quad
\|\bpi-\btheta\|_{p,\bw}
=
\|\bar\theta\boldsymbol 1-\btheta\|_{p,\bw},
\end{equation}
which does not in general imply
$\bpi=\bar\theta\boldsymbol 1$.
Finally,
\[
D_p^{\rm cpl}(\bpi;\btheta,\bw)\leq1,
\]
but the coordinate need not be non-negative outside the canonical
pooling--profiling range.
\end{proposition}

The canonical-path result follows from
\[
\bpi^{(a)}-\btheta
=
(1-a)(\bar\theta\boldsymbol 1-\btheta).
\]
The remaining properties follow directly from the weighted $L^p$ norm;
proofs are collected in Appendix~\ref{sup:main-proofs}. Notice that the profiling
endpoint is identified by the value one, whereas the value zero does not
uniquely identify the pooling tariff away from the canonical path.

\subsection{The $L^1$ transfer interpretation}
\label{subsec:transfer}

For $p=1$, balance turns the coupled distance into an exact measure of benchmark-relative transfers. Define the total positive transfer volume
\begin{equation}
\label{eq:transfer-volume}
T(\bpi;\btheta,\bw)
=
\sum_{i=1}^n w_i(\pi_i-\theta_i)_+.
\end{equation}
Policies with $\pi_i>\theta_i$ contribute more than their benchmark pure premium, whereas those with $\pi_i<\theta_i$ contribute less. Under balance, the positive and negative parts have the same weighted total.

\begin{proposition}[Transfer interpretation of the coupled $L^1$ coordinate]
\label{prop:transfer-coupled-l1}
Let $\btheta$ be non-degenerate and let $\bpi$ satisfy \eqref{eq:balance}. Then
\[
T(\bpi;\btheta,\bw)
=
\frac12\|\bpi-\btheta\|_{1,\bw},
\]
and hence
\begin{equation}
\label{eq:transfer-interpretation}
D_1^{\rm cpl}(\bpi;\btheta,\bw)
=
1-
\frac{T(\bpi;\btheta,\bw)}
{T(\bpi^{\rm pool};\btheta,\bw)}.
\end{equation}
\end{proposition}

Thus $D_1^{\rm cpl}$ is the fraction of the full-pooling benchmark-relative transfer volume removed by the evaluated tariff. Equivalently, $1-D_1^{\rm cpl}$ is the fraction that remains. This identity holds for every balanced tariff, not only on the canonical path. It gives the $L^1$ coordinate a direct actuarial interpretation in monetary units before normalization and as a share after normalization.

\subsection{Marginal Wasserstein coordinate}
\label{subsec:wasserstein}

The second coordinate asks a different question: does the tariff reproduce the benchmark's marginal monetary distribution, regardless of which policy receives which premium? For any portfolio vector $\bx=(x_1,\ldots,x_n)$, define its exposure-weighted empirical distribution
\[
\mu_{\bx}=\sum_{i=1}^n w_i\delta_{x_i}.
\]
The weights have a direct exposure interpretation: drawing a unit of exposure at random from the portfolio selects policy $i$ with probability $w_i$.

For probability measures $\mu$ and $\nu$ on $\mathbb R$ with finite $p$th moments, the $p$-Wasserstein distance is
\begin{equation}
\label{eq:wasserstein-definition}
W_p(\mu,\nu)
=
\left[
\inf_{\gamma\in\Pi(\mu,\nu)}
\int_{\mathbb R^2}|x-y|^p\,d\gamma(x,y)
\right]^{1/p},
\end{equation}
where $\Pi(\mu,\nu)$ is the set of probability measures on $\mathbb R^2$ with marginals $\mu$ and $\nu$. In one dimension, if
\[
Q_\mu(u)=\inf\{x:F_\mu(x)\geq u\},
\qquad 0<u<1,
\]
denotes the generalized quantile function, then
\begin{equation}
\label{eq:wasserstein-quantile}
W_p^p(\mu,\nu)
=
\int_0^1|Q_\mu(u)-Q_\nu(u)|^p\,du
\end{equation}
\citep{Santambrogio2015,Villani2009,PeyreCuturi2019,PanaretosZemel2020}.

Define the \emph{marginal Wasserstein coordinate}
\begin{equation}
\label{eq:wasserstein-coordinate}
D_p^W(\bpi;\btheta,\bw)
=
1-
\frac{W_p(\mu_{\bpi},\mu_{\btheta})}
{W_p(\delta_{\bar\theta},\mu_{\btheta})}.
\end{equation}

\begin{proposition}[Properties of the marginal Wasserstein profiling coordinate]
\label{prop:wasserstein-properties}
Let $p\geq1$ and let $\btheta$ be non-degenerate. Then
$D_p^W$ satisfies:

\begin{itemize}
    \item[-] endpoint normalization,
    \[
    D_p^W(\bar\theta\boldsymbol 1;\btheta,\bw)=0,
    \qquad
    D_p^W(\btheta;\btheta,\bw)=1;
    \]

    \item[-] canonical-path calibration,
    \[
    D_p^W\bigl(
    (1-a)\bar\theta\boldsymbol 1+a\btheta;
    \btheta,\bw
    \bigr)
    =a,
    \qquad 0\leq a\leq1;
    \]

    \item[-] invariance under a common positive monetary rescaling;

    \item[-] invariance under joint relabelling of policy values and
    exposure weights;

    \item[-] continuity away from benchmark degeneracy.
\end{itemize}

Moreover,
\begin{equation}
\label{eq:wasserstein-one-characterization}
D_p^W(\bpi;\btheta,\bw)=1
\quad\Longleftrightarrow\quad
\mu_{\bpi}=\mu_{\btheta}.
\end{equation}
Thus the upper endpoint uniquely characterizes equality of the
exposure-weighted marginal distributions, but not policy-level profiling.

By contrast,
\begin{equation}
\label{eq:wasserstein-zero-characterization}
D_p^W(\bpi;\btheta,\bw)=0
\quad\Longleftrightarrow\quad
W_p(\mu_{\bpi},\mu_{\btheta})
=
W_p(\delta_{\bar\theta},\mu_{\btheta}),
\end{equation}
which does not in general imply
$\mu_{\bpi}=\delta_{\bar\theta}$.
Finally,
\[
D_p^W(\bpi;\btheta,\bw)\leq1,
\]
but the coordinate need not be non-negative outside the canonical
pooling--profiling range.
\end{proposition}

The canonical-path identity follows from
\[
Q_{\mu_{\bpi^{(a)}}}(u)
=
(1-a)\bar\theta+aQ_{\mu_{\btheta}}(u),
\]
which gives
\[
W_p(\mu_{\bpi^{(a)}},\mu_{\btheta})
=
(1-a)W_p(\delta_{\bar\theta},\mu_{\btheta}).
\]
Unlike the coupled characterization
\eqref{eq:coupled-one-characterization}, equation~\eqref{eq:wasserstein-one-characterization} does not require the benchmark premium values to be assigned to the same policies.

\subsection{Allocation mismatch and values outside the canonical range}
\label{subsec:mismatch}

Pairing $\pi_i$ with $\theta_i$ policy by policy gives one admissible transport plan between $\mu_{\bpi}$ and $\mu_{\btheta}$. Wasserstein distance chooses the least costly plan. This immediately orders the two coordinates.

\begin{proposition}[Coupled and marginal profiling]
\label{prop:mismatch}
For every $p\geq1$,
\[
W_p(\mu_{\bpi},\mu_{\btheta})
\leq
\|\bpi-\btheta\|_{p,\bw},
\]
and therefore
\begin{equation}
\label{eq:ordering-coordinates}
D_p^W(\bpi;\btheta,\bw)
\geq
D_p^{\rm cpl}(\bpi;\btheta,\bw).
\end{equation}
Define the normalized allocation mismatch
\begin{equation}
\label{eq:mismatch}
M_p(\bpi;\btheta,\bw)
=
\frac{
\|\bpi-\btheta\|_{p,\bw}^p
-
W_p^p(\mu_{\bpi},\mu_{\btheta})
}
{\|\bar\theta\boldsymbol 1-\btheta\|_{p,\bw}^p}.
\end{equation}
Then $M_p\geq0$ and
\begin{equation}
\label{eq:mismatch-identity}
M_p(\bpi;\btheta,\bw)
=
\{1-D_p^{\rm cpl}(\bpi;\btheta,\bw)\}^p
-
\{1-D_p^W(\bpi;\btheta,\bw)\}^p.
\end{equation}
In particular,
\[
M_1(\bpi;\btheta,\bw)
=
D_1^W(\bpi;\btheta,\bw)-D_1^{\rm cpl}(\bpi;\btheta,\bw).
\]
Moreover, $M_p=0$ exactly when the observed policy-level pairing attains the Wasserstein minimum.
\end{proposition}

The three quantities therefore have distinct roles. $D_p^{\rm cpl}$ measures policy-level alignment; $D_p^W$ measures marginal monetary alignment; and $M_p$ is the normalized reduction in $p$-cost obtainable by optimally reallocating premium values while keeping the two marginal distributions fixed. For $p=1$, this cost mismatch is numerically the vertical gap between the two coordinates. For $p>1$, it is not simply their difference and need not lie in $[0,1]$.

The mismatch vanishes at both endpoints for different reasons. At full pooling there is no premium dispersion to allocate incorrectly, so the policy-level pairing is automatically optimal. At full profiling, every policy receives its benchmark value. Between these endpoints, a positive mismatch means that some of the tariff's disagreement with the benchmark is attributable to which policies receive which premium values rather than to the marginal distribution itself.

\begin{example}[Right distribution, wrong policies]
\label{ex:reversal}
Let $n=2$, $w_1=w_2=1/2$, and
\[
\btheta=(80,120),
\qquad
\bpi^{\rm pool}=(100,100).
\]
Consider the balanced reversed tariff
\[
\bpi=(120,80).
\]
Because $\mu_{\bpi}=\mu_{\btheta}$, the marginal coordinate is $D_1^W=1$. In contrast,
\[
\|\bpi-\btheta\|_{1,\bw}=40,
\qquad
\|\bpi^{\rm pool}-\btheta\|_{1,\bw}=20,
\]
so
\[
D_1^{\rm cpl}=-1,
\qquad
M_1=2.
\]
The negative coupled value does not mean ``more pooled than full pooling''. It means that the policy-level tariff is farther from the profiling benchmark than the full-pooling tariff. More generally, both coordinates are bounded above by one but can be negative for tariffs that lie beyond the pooling endpoint in the relevant discrepancy. The interval $[0,1]$ is therefore the canonical pooling--profiling range, not the mathematical range of every possible tariff.
\end{example}

The coordinates are defined on finite portfolios, but they have natural population limits. Under standard moment and non-degeneracy conditions, the marginal limit depends only on the marginal laws of premium and benchmark risk, whereas the coupled limit depends on their joint law. Almost-sure and plug-in consistency results are given in Appendices~\ref{app:population-limits} and~\ref{app:estimated-benchmark}.

\section{From the coordinates to an actuarial audit}
\label{sec:audit}

\subsection{Benchmark construction, balancing, and estimation}
\label{subsec:benchmark-construction}

In a synthetic experiment the benchmark may be known. In real data it is estimated. A natural target for a technical pure-premium benchmark is
\[
\theta(x)=\mathbb E(Y\mid X=x),
\]
where $X$ denotes the covariates admitted to the reference model. This notation specifies the audit benchmark; it does not assert that a unique intrinsic individual risk is observable.

When the benchmark is fitted on the same portfolio used for evaluation, out-of-sample construction limits optimistic alignment. With $K$-fold cross-fitting, let $k(i)$ denote the fold containing policy $i$, fit the benchmark learner without that fold, and set
\[
\widehat\theta_i
=
\widehat m^{(-k(i))}(X_i).
\]
The vector $\widehat\btheta=(\widehat\theta_1,\ldots,\widehat\theta_n)$ is then fully out of fold. The same principle applies to candidate tariffs estimated from the data: their evaluated premium vectors should be generated out of sample whenever the empirical design permits it \citep{Bach2024,CuckerZhou2007,WuthrichMerz2023}.

Let $\bpi^{\rm raw}$ denote a candidate premium vector before balancing and write
\[
\bar\pi^{\rm raw}=\sum_iw_i\pi_i^{\rm raw}.
\]
For positive technical premiums, our default is multiplicative balancing,
\begin{equation}
\label{eq:multiplicative-balance}
\pi_i
=
\pi_i^{\rm raw}\frac{\bar\theta}{\bar\pi^{\rm raw}},
\end{equation}
which preserves premium ratios. Additive balancing,
\[
\pi_i=\pi_i^{\rm raw}-\bar\pi^{\rm raw}+\bar\theta,
\]
preserves absolute monetary differences but can produce negative premiums. We therefore use it only as a sensitivity check when positivity is preserved.

The main analysis concerns technical pure premiums. A commercial premium may also contain expenses, commissions, capital or profit loadings, taxes and market adjustments. Applying the same coordinates to a balanced commercial tariff is possible, but then the measured differentiation combines risk classification with these additional pricing layers. Appendix~\ref{app:commercial-premiums} discusses this distinction and a corresponding commercial re-profiling diagnostic.

Because the scale is benchmark-relative, the reference should be reported transparently and varied when plausible alternatives exist. A benchmark ladder such as
\[
\widehat\btheta^{\rm GLM},
\qquad
\widehat\btheta^{\rm GAM},
\qquad
\widehat\btheta^{\rm XGB}
\]
shows which conclusions survive changes in the accepted profiling frontier. The coordinates are stable to small perturbations of a non-degenerate benchmark; finite-portfolio bounds and plug-in consistency results are given in Appendices~\ref{sup:benchmark-stability} and~\ref{app:estimated-benchmark}. Sampling uncertainty can be studied by portfolio bootstrap, full refit bootstrap or repeated cross-fitting. In the motor application, we use a benchmark ladder in the main text and report a conditional multiplier-bootstrap analysis in Appendix~\ref{app:bootstrap-uncertainty}.

\subsection{Gini dispersion as a companion profiling measure}
\label{subsec:gini}

The classical Gini coefficient provides a familiar actuarial summary of
premium dispersion \citep{FreesMeyersCummings2011,FreesMeyersCummings2013}.
It is naturally related to the pooling--profiling scale, although it contains
less information than either of the two distance-based coordinates.

For a non-negative portfolio vector
$\bx=(x_1,\ldots,x_n)$ with positive exposure-weighted mean
$\bar x=\sum_i w_i x_i$, define the exposure-weighted Gini coefficient by
\begin{equation}
\label{eq:weighted-gini}
\mathcal G_{\bw}(\bx)
=
\frac{
\sum_i\sum_j w_iw_j|x_i-x_j|
}{
2\bar x
}.
\end{equation}
For a tariff $\bpi$ balanced to the benchmark mean $\bar\theta$, define the
benchmark-relative Gini ratio
\begin{equation}
\label{eq:gini-ratio}
D^{\rm G}(\bpi;\btheta,\bw)
=
\frac{\mathcal G_{\bw}(\bpi)}
     {\mathcal G_{\bw}(\btheta)},
\end{equation}
whenever the benchmark is non-degenerate.

We call $D^{\rm G}$ a ratio rather than a third pooling--profiling
coordinate. It records how much of the benchmark's relative dispersion is
present in the tariff, but does not measure the full discrepancy between the
two marginal distributions and contains no information about which policies
receive which premium values. It nevertheless shares several structural
properties with the two main coordinates.

\begin{proposition}[Properties of the benchmark-relative Gini ratio]
\label{prop:gini-properties}
Let $\btheta$ be non-constant with $\bar\theta>0$, and let $\bpi$ be a
non-negative tariff balanced to the same mean. Then $D^{\rm G}$ satisfies:

\begin{itemize}
    \item[-] endpoint normalization,
    \[
    D^{\rm G}(\bar\theta\boldsymbol 1;\btheta,\bw)=0,
    \qquad
    D^{\rm G}(\btheta;\btheta,\bw)=1;
    \]

    \item[-] canonical-path calibration,
    \[
    D^{\rm G}\bigl(
    (1-a)\bar\theta\boldsymbol 1+a\btheta;
    \btheta,\bw
    \bigr)
    =a,
    \qquad 0\leq a\leq1;
    \]

    \item[-] invariance under a common positive monetary rescaling;

    \item[-] invariance under joint relabelling of portfolio values and
    exposure weights;

    \item[-] continuity away from benchmark degeneracy.
\end{itemize}

Moreover,
\begin{equation}
\label{eq:gini-zero-characterization}
D^{\rm G}(\bpi;\btheta,\bw)=0
\quad\Longleftrightarrow\quad
\bpi=\bar\theta\boldsymbol 1
\end{equation}
for policies with positive exposure weight, whereas
\begin{equation}
\label{eq:gini-one-characterization}
D^{\rm G}(\bpi;\btheta,\bw)=1
\quad\Longleftrightarrow\quad
\mathcal G_{\bw}(\bpi)=\mathcal G_{\bw}(\btheta).
\end{equation}
The latter condition does not imply either
$\mu_{\bpi}=\mu_{\btheta}$ or $\bpi=\btheta$.

Finally,
\begin{equation}
\label{eq:gini-wasserstein-bound}
\left|
\mathcal G_{\bw}(\bpi)
-
\mathcal G_{\bw}(\btheta)
\right|
\leq
\frac{
W_1(\mu_{\bpi},\mu_{\btheta})
}{
\bar\theta
}.
\end{equation}
Equivalently,
\begin{equation}
\label{eq:gini-coordinate-bound}
\frac{
|D^{\rm G}(\bpi;\btheta,\bw)-1|
}{
\kappa_\theta
}
\leq
1-D_1^W(\bpi;\btheta,\bw)
\leq
1-D_1^{\rm cpl}(\bpi;\btheta,\bw),
\end{equation}
where
\begin{equation}
\label{eq:kappa-gini}
\kappa_\theta
=
\frac{
W_1(\delta_{\bar\theta},\mu_{\btheta})
}{
\bar\theta\,\mathcal G_{\bw}(\btheta)
}
\in[1,2].
\end{equation}
\end{proposition}

The proof is given in Appendix~\ref{sup:canonical-diagnostics}. The canonical-path identity is
particularly useful. Along
\[
\bpi^{(a)}
=
(1-a)\bar\theta\boldsymbol 1+a\btheta,
\]
the three quantities satisfy
\begin{equation}
\label{eq:gini-three-calibrations}
D^{\rm G}(\bpi^{(a)})
=
D_1^W(\bpi^{(a)})
=
D_1^{\rm cpl}(\bpi^{(a)})
=
a.
\end{equation}
Thus $a$ can simultaneously be read as the fraction of benchmark Gini
dispersion reproduced by the tariff, the degree of marginal Wasserstein
profiling, and the degree of policy-level coupled profiling. Away from the
canonical path, however, these three readings separate.

The characterization of their upper endpoints makes the information loss
explicit. For the coupled coordinate,
\[
D_1^{\rm cpl}=1
\quad\Longleftrightarrow\quad
\bpi=\btheta
\]
policy by policy. For the marginal coordinate,
\[
D_1^W=1
\quad\Longleftrightarrow\quad
\mu_{\bpi}=\mu_{\btheta}.
\]
For the Gini ratio, by contrast, $D^{\rm G}=1$ requires only equality of
one scalar dispersion functional. Thus the three conditions successively
discard policy-level allocation and then marginal-distribution information.

Example~\ref{ex:reversal} illustrates the first loss particularly sharply.
For $\btheta=(80,120)$ and the reversed tariff $\bpi=(120,80)$, the two
marginal distributions coincide, so
\[
D^{\rm G}=D_1^W=1,
\]
whereas
\[
D_1^{\rm cpl}=-1,
\qquad
M_1=2.
\]
Neither the classical Gini nor any other purely marginal statistic can
detect this allocation reversal.

Equation~\eqref{eq:gini-wasserstein-bound} also clarifies the relation
between the Gini ratio and the marginal Wasserstein coordinate. A tariff
that is close to the benchmark in $W_1$ must have a Gini coefficient close
to the benchmark Gini. The converse is false: equality of Gini coefficients
places only a scalar restriction on the marginal distributions. In this
sense, $D_1^W$ refines the classical Gini measure by retaining the entire
monetary distribution.

Finally, $D^{\rm G}$ need not lie in $[0,1]$. A balanced tariff can be more
dispersed than the benchmark and therefore have $D^{\rm G}>1$. The interval
$[0,1]$ is the canonical pooling--profiling range for the Gini ratio, just
as it is for the two distance-based coordinates, but not its full
mathematical range. Residual ordered-Gini and Lorenz diagnostics, which
address remaining segmentation opportunities rather than marginal premium
dispersion, are reported in the appendices.

\subsection{Other complementary actuarial diagnostics}
\label{subsec:complementary-diagnostics}

The pooling--profiling coordinates answer a specific question and should be reported alongside, not instead of, standard actuarial diagnostics. Group transfers show who contributes to the benchmark-relative redistribution that remains. For a group attribute $S_i$ and group $s$, define
\begin{equation}
\label{eq:group-transfer}
\Delta_s(\bpi;\btheta,\bw)
=
\frac{\sum_{i:S_i=s}w_i(\pi_i-\theta_i)}
{\sum_{i:S_i=s}w_i}.
\end{equation}
Positive values identify net contributors relative to the profiling benchmark and negative values net beneficiaries. Under portfolio balance, the group transfers average to zero with group exposure shares as weights. If ability-to-pay information is available, premium burden or affordability should be reported separately because access is a different outcome from demutualization.

Gini and ordered-Lorenz diagnostics are complementary for the same reason. Classical Gini measures premium inequality, while ordered Lorenz curves assess ranking and residual segmentation opportunities \citep{FreesMeyersCummings2011,FreesMeyersCummings2013}. Rank-based measures are invariant to monotone transformations and therefore cannot distinguish, for example, a benchmark tariff from a strong shrinkage of that tariff toward the pooled mean. Additional Gini and residual-Lorenz results are reported in the appendices.

Calibration also answers a different question. Let $I$ denote a policy drawn with $\mathbb P(I=i)=w_i$, and define $\Pi=\pi_I$ and $\Theta=\theta_I$. Calibration relative to the benchmark means
\[
\mathbb E(\Theta\mid\Pi)=\Pi.
\]
Both endpoints can satisfy this condition: under full pooling, $\Pi=\bar\theta$ and $\mathbb E(\Theta\mid\Pi)=\bar\theta$; under full profiling, $\Pi=\Theta$. Calibration therefore cannot identify how much benchmark heterogeneity has been monetized into premium differences \citep{DenuitHainautTrufin2019GLM,WuthrichMerz2023}.

A minimum audit should therefore report the benchmark construction, the premium layer being evaluated, the balancing rule, and at least $D_1^{\rm cpl}$, $D_1^W$ and $M_1$. The $p=2$ versions are useful when large policy-level deviations deserve greater weight. Portfolio coordinates should be supplemented by the group-transfer, disparity or affordability diagnostics relevant to the application, together with benchmark sensitivity and an uncertainty assessment.

\subsection{Group parity as a pricing intervention}
\label{subsec:group-parity-mutualization}

A group-parity correction is a useful stress test because it can alter premium allocation while leaving substantial overall dispersion in place. We use a one-dimensional Wasserstein-barycenter correction as an intervention on an existing tariff \citep{FrezalBarry2019,charpentier2023mitigating,charpentier2024insurance,lindholm2024fair}.

For this population statement only, capital letters denote random variables. Let $S\in\{0,1\}$ with $p_s=\mathbb P(S=s)$, and let $Q_s$ be the quantile function of $\Pi\mid S=s$. The barycenter quantile is
\[
Q_B(u)=p_0Q_0(u)+p_1Q_1(u),
\qquad 0<u<1.
\]
If the conditional distributions are continuous, mapping $\Pi$ in group $s$ to $Q_B\{F_s(\Pi)\}$ gives both groups the same conditional premium distribution. With atoms, one can use the randomized distributional transform
\[
U_s=F_s(\Pi^-)+V\{F_s(\Pi)-F_s(\Pi^-)\},
\qquad V\sim\mathrm{Uniform}(0,1),
\]
independently of $(\Pi,S)$, and map to $Q_B(U_s)$. The empirical implementation in Section~\ref{subsec:real-group-parity} uses exposure-weighted within-group midranks.

The transfer implication is immediate. Suppose the corrected tariff is balanced and has the same conditional premium distribution in both groups. The two conditional premium means must then both equal $\bar\theta$, so
\[
\mathbb E(\Pi-\Theta\mid S=s)
=
\bar\theta-\mathbb E(\Theta\mid S=s).
\]
Equalizing conditional premium distributions therefore does not remove a difference in benchmark means between groups; it changes whether that difference appears in premiums or in benchmark-relative transfers.

Along the canonical path this mechanism yields a direct statement about the coupled coordinate. Write $\bar\theta=\mathbb E(\Theta)$ and $\|Z\|_{L^p}=\{\mathbb E|Z|^p\}^{1/p}$. The population coupled coordinate is
\[
D_{p,\infty}^{\rm cpl}(\Pi;\Theta)
=
1-
\frac{\|\Pi-\Theta\|_{L^p}}
{\|\Theta-\mathbb E(\Theta)\|_{L^p}},
\]
whenever the denominator is positive.

\begin{proposition}[Group parity on the canonical path]
\label{prop:parity-reduces-coupled}
Assume $p_s>0$ for $s=0,1$. Let $Q_s^\Theta$ be the quantile function of $\Theta\mid S=s$, and let $U_S$ be a conditional uniform rank such that $\Theta=Q_S^\Theta(U_S)$ almost surely. Define
\[
Q_B^\Theta(u)=p_0Q_0^\Theta(u)+p_1Q_1^\Theta(u).
\]
For $a\in[0,1]$, let
\[
\Pi^{(a)}=(1-a)\bar\theta+a\Theta
\]
and define its barycentric parity correction by
\[
\Pi_{\rm par}^{(a)}
=(1-a)\bar\theta+aQ_B^\Theta(U_S).
\]
Then, for every $p\geq1$,
\[
\|\Pi_{\rm par}^{(a)}-\Theta\|_{L^p}
\geq
\|\Pi^{(a)}-\Theta\|_{L^p}.
\]
Consequently, $D_{p,\infty}^{\rm cpl}$ weakly decreases after the parity correction. For $p=1$, the corrected tariff removes no more of the full-pooling benchmark-relative transfer volume than the original canonical tariff.
\end{proposition}

The proof and an exact $L^2$ decomposition are given in Appendix~\ref{sup:parity}. Proposition~\ref{prop:parity-reduces-coupled} concerns policy-level alignment, not marginal dispersion. The marginal Wasserstein coordinate may remain high after the correction because premium values can retain a benchmark-like distribution while being reassigned across policies. Section~\ref{subsec:real-group-parity} shows this mechanism in the motor portfolio.

Finally, the entire framework is a fixed-portfolio audit. Exposures are held fixed, and the coordinates do not model participation, coverage choice, retention or competitive responses. Such responses may be important consequences of a pricing rule, but they require an explicit behavioral or market model rather than a diagnostic defined on a fixed premium vector.

\section{Synthetic illustration}
\label{sec:synthetic}

We first use a controlled portfolio in which the profiled benchmark is known. The experiment checks the endpoints and canonical path, contrasts coupled and marginal profiling, and illustrates how correlated proxies can preserve substantial differentiation after sensitive variables are excluded.

\subsection{Data-generating process}
\label{subsec:dgp}

We simulate $n=50{,}000$ policyholders with equal exposure $e_i=1$, hence $w_i=1/n$. Two binary sensitive attributes are generated as
\[
A_i\sim\mathrm{Bernoulli}(0.5),
\qquad
B_i\mid A_i
\sim
\mathrm{Bernoulli}\{\operatorname{logit}^{-1}(-0.4+1.1A_i)\}.
\]
Four rating variables are then generated. $X_1$ is an age-like variable, $X_2$ is a continuous risk characteristic correlated with the sensitive attributes, and $X_3$ and $X_4$ are constructed as increasingly strong proxy-type covariates:
\[
X_{1i}=40+10A_i+6B_i+\varepsilon_{1i},
\]
truncated to $[18,80]$,
\[
X_{2i}=0.6A_i+0.3B_i+\varepsilon_{2i},
\qquad
X_{3i}=0.8B_i+0.3A_i+\varepsilon_{3i},
\]
and
\[
X_{4i}=0.5X_{3i}+0.4B_i+\varepsilon_{4i}.
\]
The noise terms are independent centered Gaussian variables, and $X_2,X_3,X_4$ are standardized. The labels are deliberately generic because the experiment is designed to isolate pricing mechanisms rather than reproduce a specific insurance product; $X_3$ and $X_4$ act as strong proxies for $A$ and $B$.

Claim counts follow
\[
N_i\mid \bX_i,A_i,B_i\sim\mathrm{Poisson}(\lambda_i),
\]
with
\[
\log\lambda_i
=
-2.80
+0.010(X_{1i}-40)
+0.25X_{2i}
+0.35X_{3i}
+0.45X_{4i}
+0.20A_i
+0.10B_i
+0.25X_{3i}B_i.
\]
Claim severity is fixed at $c_0=1000$. Hence $Y_i=c_0N_i$ and the oracle pure premium is
\[
\theta_i=\mathbb E(Y_i\mid \bX_i,A_i,B_i)=c_0\lambda_i.
\]

\subsection{Pricing regimes and metrics}
\label{subsec:pricing-regimes}

We compare eight regimes, all balanced to the mean of $\btheta$. The pooled tariff is $\pi_i^{\mathrm{pool}}=\bar\theta$. The coarse tariff charges empirical means within broad bands formed from $X_1$, $X_2$ and $X_3$. The proxy GLM uses $X_1,\ldots,X_4$ but not $A$ or $B$; the full GLM adds both sensitive variables and the interaction $X_3B$. We also consider
\[
\pi_i^{\mathrm{shrink}}
=(1-\rho)\bar\theta+\rho\pi_i^{\mathrm{full}},
\qquad \rho=0.60,
\]
and the capped tariff
\[
\pi_i^{\mathrm{cap}}
=\min\{2\bar\theta,\max(0.5\bar\theta,\pi_i^{\mathrm{full}})\}.
\]
The no-proxy GLM uses only $X_1$ and $X_2$, while $\pi_i^{\mathrm{oracle}}=\theta_i$ defines the profiling endpoint.

For each tariff we report the coupled and marginal coordinates, the allocation mismatch, and selected dispersion and tail summaries. By Proposition~\ref{prop:transfer-coupled-l1}, $D_1^{\mathrm{cpl}}$ is also the fraction of pooling-induced transfer volume removed. Figures~\ref{fig:synthetic-quantiles} and~\ref{fig:synthetic-transfers} and Table~\ref{tab:synthetic-metrics-compact} summarize the main results. Residual ordered-Lorenz and complete Gini diagnostics are reported in the appendices.

\begin{figure}[t]
\centering
\includegraphics[width=.82\textwidth]{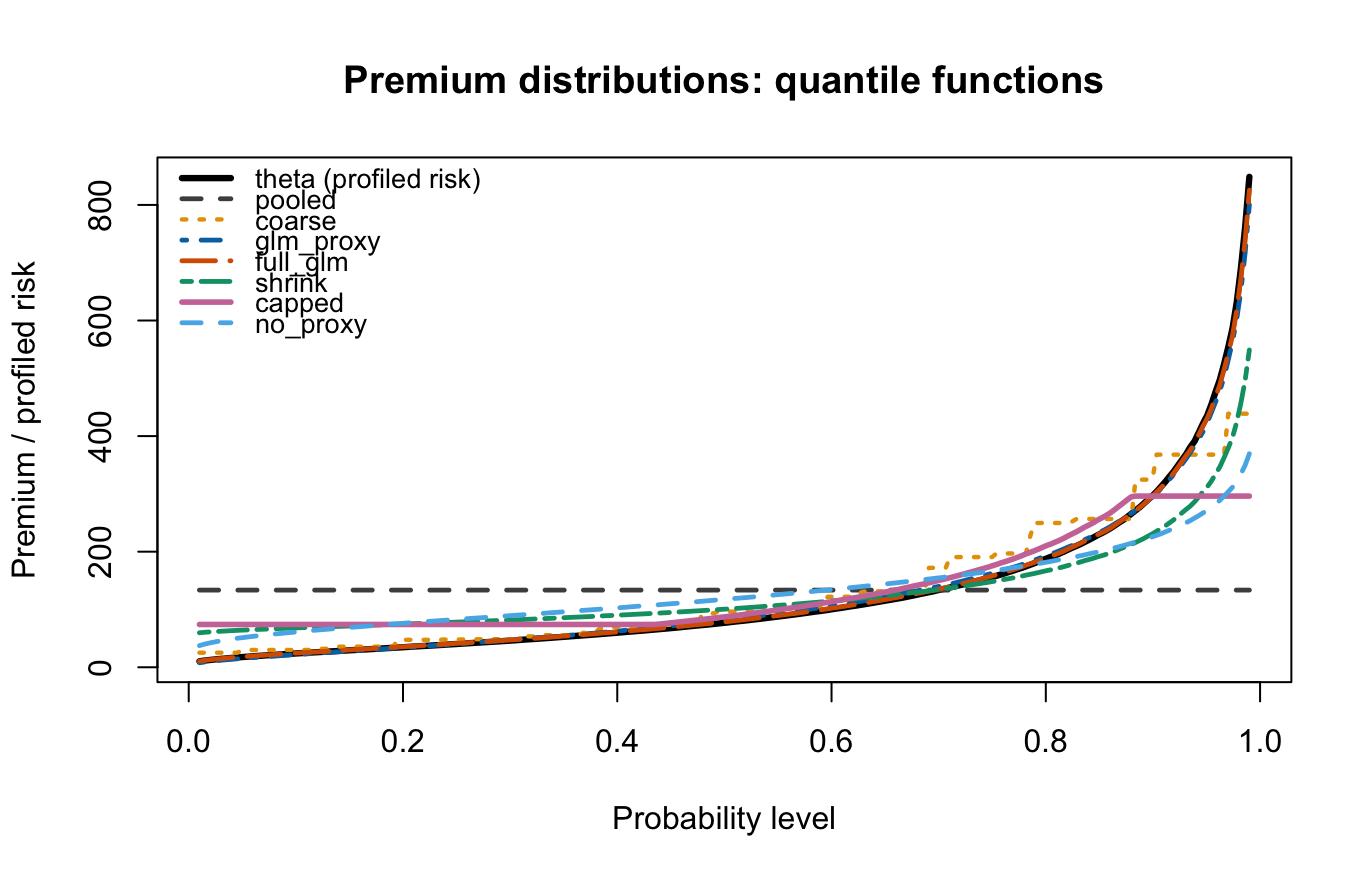}
\caption{Empirical premium quantile functions. Shrinkage and capping compress the benchmark distribution toward the pooled mean.}
\label{fig:synthetic-quantiles}
\end{figure}

\begin{figure}[t]
\centering
\includegraphics[width=.82\textwidth]{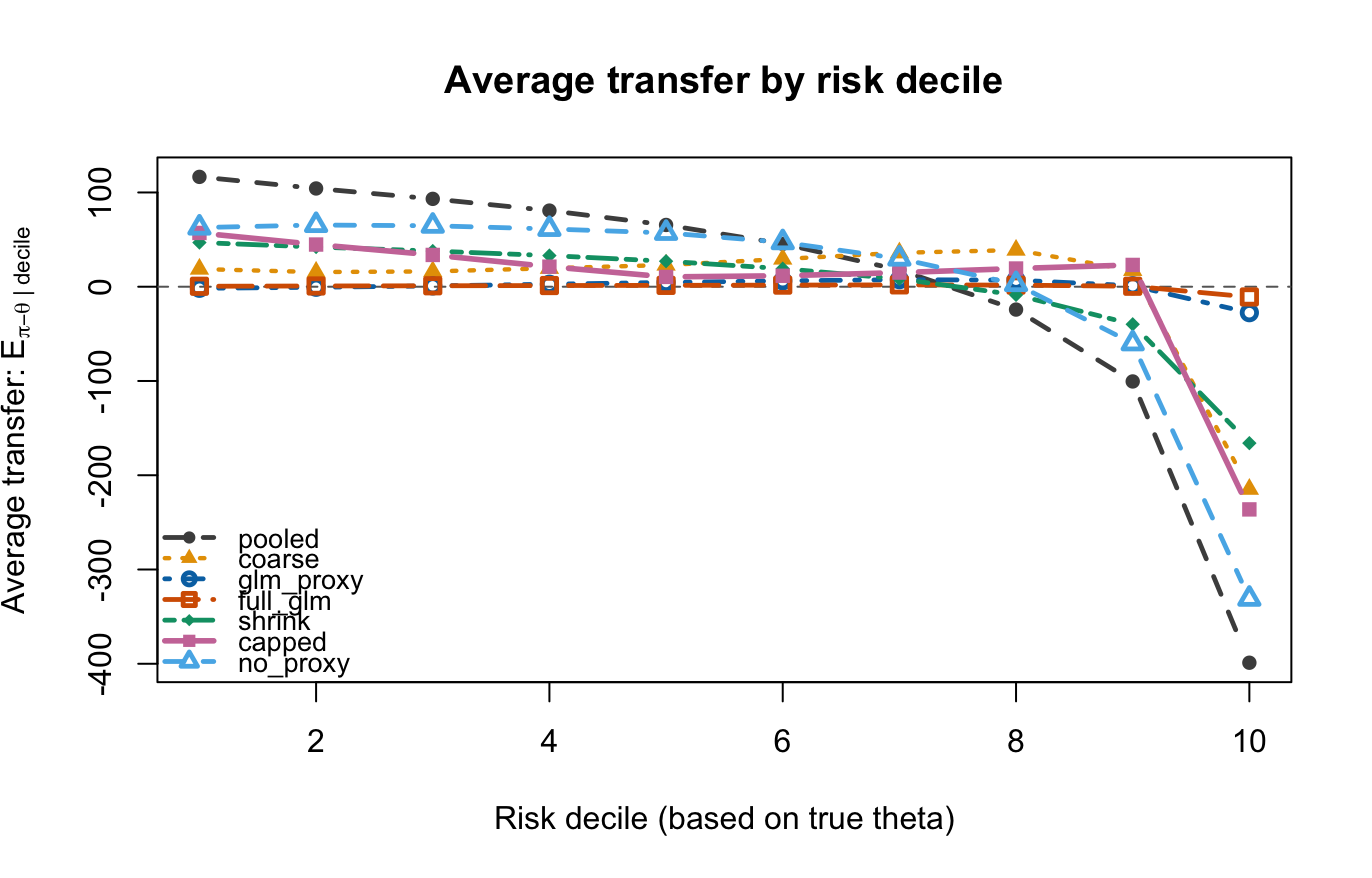}
\caption{Average benchmark-relative transfers by risk decile. Full pooling redistributes from low- to high-risk deciles, whereas the full GLM nearly eliminates these transfers.}
\label{fig:synthetic-transfers}
\end{figure}

\begin{table}[t]
\centering
\caption{Pooling--profiling metrics for the synthetic portfolio. The exact identity $\overline M_1=D_1^W-D_1^{\mathrm{cpl}}$ gives the normalized $L^1$ allocation mismatch; $q_{.95}$ is the 95th premium quantile.}
\label{tab:synthetic-metrics-compact}
\small
\begin{tabular}{lrrrrrrr}
\toprule
Pricing rule
& $D_1^{\mathrm{cpl}}$
& $D_2^{\mathrm{cpl}}$
& $D_1^W$
& $D_2^W$
& $\overline M_1$
& $\operatorname{sd}(\bpi)$
& $q_{.95}$ \\
\midrule
Pooled          & 0.000 & 0.000 & 0.000 & 0.000 & 0.000 &   0.000 & 133.539 \\
Coarse tariff   & 0.392 & 0.238 & 0.703 & 0.423 & 0.311 & 113.041 & 367.923 \\
GLM with proxies& 0.863 & 0.836 & 0.950 & 0.927 & 0.087 & 165.197 & 420.717 \\
Full GLM        & 0.961 & 0.949 & 0.979 & 0.965 & 0.018 & 170.208 & 428.463 \\
Shrinkage       & 0.590 & 0.581 & 0.591 & 0.582 & 0.001 & 102.180 & 310.589 \\
Capped tariff   & 0.549 & 0.292 & 0.549 & 0.292 & 0.000 &  79.639 & 296.200 \\
No-proxy GLM    & 0.145 & 0.086 & 0.473 & 0.355 & 0.328 &  70.978 & 271.291 \\
Oracle          & 1.000 & 1.000 & 1.000 & 1.000 & 0.000 & 175.394 & 433.418 \\
\bottomrule
\end{tabular}
\end{table}

\subsection{Results}
\label{subsec:synthetic-results}

The pooled and oracle tariffs attain the normalized endpoints. The shrinkage tariff provides a more informative calibration check: with $\rho=0.60$, the four principal values range from $0.581$ to $0.591$, close to the imposed location on the canonical path. Small deviations arise because shrinkage is applied to the fitted full GLM rather than directly to the oracle benchmark.

The coupled and marginal views separate sharply for coarse and no-proxy tariffs. The coarse tariff has $D_1^{\mathrm{cpl}}=0.392$ but $D_1^W=0.703$: its broad cells create substantial portfolio-level dispersion while leaving considerable allocation error within cells. The no-proxy GLM displays an even larger gap, $0.145$ versus $0.473$. By contrast, the proxy GLM remains close to profiling despite excluding $A$ and $B$, because $X_3$ and $X_4$ recover much of their risk information. Excluding sensitive variables therefore does not mechanically restore pooling when strong proxies remain.

The capped tariff shows why both $p=1$ and $p=2$ are useful. Its profiling level is about $0.55$ in $L^1$ but only $0.29$ in $L^2$, because capping acts mainly in the tails. Shrinkage, by comparison, preserves ranking while compressing monetary differences throughout the distribution. The comparison illustrates the different loss functions: $p=2$ places more weight on large policy-level deviations than $p=1$. Standard Gini and residual-segmentation diagnostics, reported in Appendices~\ref{sup:canonical-diagnostics} and~\ref{sup:synthetic}, do not provide this same monetary calibration.

Figures~\ref{fig:synthetic-quantiles} and~\ref{fig:synthetic-transfers} give the distributional and redistributive readings. The proxy and full GLM quantiles closely follow the benchmark, the capped tariff flattens its upper tail, and pooling creates the largest transfers across risk deciles. The experiment therefore motivates the reporting convention used below: coupled and marginal coordinates should be shown together, with transfer and tail diagnostics used to explain why they differ.

\section{Real-data application}
\label{sec:real-data}

We apply the pooling--profiling framework to the \texttt{ausprivauto0405} motor portfolio from \texttt{CASdatasets} \citep{DutangCharpentier2024}. The database contains claim counts, claim amounts, exposure, vehicle and driver characteristics, and the binary variable \texttt{Gender}. Because observed commercial premiums are not available, the exercise compares technical pure-premium rules on a fixed portfolio. Exposures are used as weights throughout. Predicted claim frequencies are converted into annual pure premiums using the empirical average claim severity. We use \texttt{Gender} as the grouping attribute $S$ in the parity intervention.

\subsection{Empirical design and reference benchmarks}
\label{subsec:real-design}

All premiums are produced out of sample by five-fold cross-fitting. The main profiling benchmark, denoted $\widehat\btheta$, is a cross-fitted XGBoost frequency model using \texttt{VehValue}, \texttt{VehAge}, \texttt{VehBody}, \texttt{DrivAge} and \texttt{Gender}. Its specification is deliberately more flexible than that of the candidate boosted-tree tariffs: the benchmark uses 600 boosting rounds, maximum depth 4 and learning rate $0.03$, whereas the candidate XGBoost tariffs use 300 rounds, maximum depth 3 and learning rate $0.05$. Both use subsampling and column subsampling rates of $0.8$. Seeds are fixed fold by fold, and the complete fitting specification is provided with the replication code. These hyperparameters are not claimed to be uniquely optimal; they define a reproducible and deliberately flexible reference whose influence is examined through the benchmark ladder in Section~\ref{subsec:real-benchmark-sensitivity}. The benchmark is therefore an operational profiling frontier, not an observed individual risk or a model-free oracle. To assess finite-portfolio uncertainty conditional on the fitted benchmark
and tariff vectors, we use a conditional exponential multiplier bootstrap
with 1,000 replications. Percentile intervals for the principal
pooling--profiling coordinates and for the paired group-parity effects are
reported in the appendices.

We compare seven rules. The \emph{pooled} tariff is constant and equal to the exposure-weighted mean of $\widehat\btheta$. The GLM and XGBoost tariffs are fitted either without $S$ or with $S$, producing \texttt{glm\_noS}, \texttt{glm\_S}, \texttt{xgb\_noS} and \texttt{xgb\_S}. The benchmark endpoint is $\bpi=\widehat\btheta$. Finally, the canonical pooling--profiling rule
\[
\pi_i^{\mathrm{shrink}}
=(1-a)\bar\theta+a\widehat\theta_i,
\qquad a=0.60,
\]
provides a direct calibration check. Every rule is multiplicatively balanced to the exposure-weighted mean of the benchmark before the coordinates are computed. Additive balancing is considered as a sensitivity analysis below.

The principal quantities are the coupled coordinate $D_1^{\mathrm{cpl}}$, the marginal Wasserstein coordinate $D_1^W$, and their exact difference
\[
M_1
=
D_1^W-D_1^{\mathrm{cpl}},
\]
which is the normalized $L^1$ allocation mismatch from Proposition~\ref{prop:mismatch}. Throughout the empirical tables, $M_1$ is computed from the unrounded coordinates; discrepancies of $0.001$ may therefore appear when the three displayed entries are compared after rounding. We also report three group diagnostics. The first is the $W_2$ distance between the two gender-conditional premium distributions. The second is the absolute difference between group-average premiums. The third is the range of group-average benchmark-relative transfers,
\[
\Delta^{\mathrm{tr}}(\bpi)
=
\max_s \Delta_s(\bpi;\widehat\btheta,\bw)
-
\min_s \Delta_s(\bpi;\widehat\btheta,\bw),
\]
with $\Delta_s$ defined in \eqref{eq:group-transfer}. For binary $S$, this is simply the absolute difference between the two group-average transfers.

\begin{table}[t]
\centering
\small
\caption{Pooling--profiling and group diagnostics for the original tariffs. The reference is the cross-fitted XGBoost benchmark. The normalized allocation mismatch satisfies $\overline M_1=D_1^W-D_1^{\mathrm{cpl}}$.}
\label{tab:cas-metrics-compact}
\begin{tabular}{lrrrrrr}
\toprule
Tariff & $D_1^{\mathrm{cpl}}$ & $D_1^W$ & $\overline M_1$ & $W_2$ gap & Premium gap & Transfer gap \\
\midrule
Pooled              & 0.000 & 0.000 & 0.000 & 0.000  & 0.000  & 9.468 \\
GLM, no $S$         & 0.329 & 0.887 & 0.559 & 9.040  & 4.590  & 4.878 \\
GLM, with $S$       & 0.332 & 0.891 & 0.558 & 12.753 & 10.696 & 1.228 \\
XGBoost, no $S$     & 0.661 & 0.808 & 0.147 & 6.905  & 1.968  & 7.500 \\
XGBoost, with $S$   & 0.677 & 0.809 & 0.132 & 10.457 & 7.672  & 1.796 \\
Benchmark shrinkage & 0.600 & 0.600 & 0.000 & 8.760  & 5.681  & 3.787 \\
Benchmark           & 1.000 & 1.000 & 0.000 & 14.601 & 9.468  & 0.000 \\
\bottomrule
\end{tabular}
\end{table}

\begin{figure}[t]
\centering
\includegraphics[width=.82\textwidth]{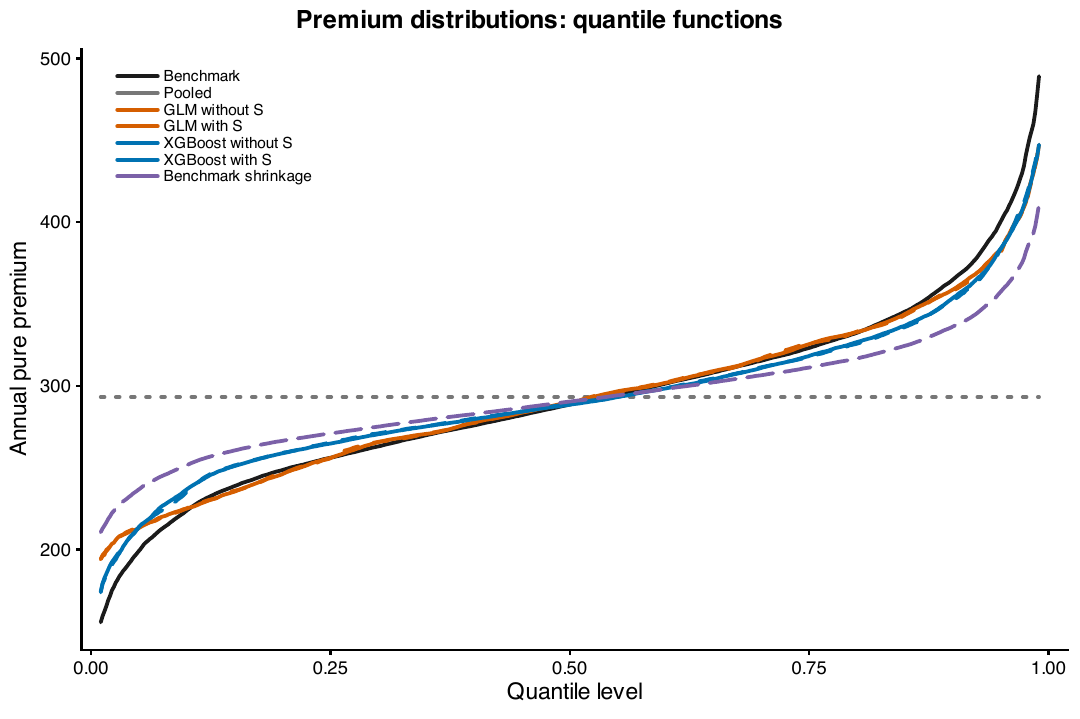}
\caption{Exposure-weighted premium quantile functions. The figure gives the marginal view of profiling: it shows how much premium heterogeneity each rule produces, but not whether that heterogeneity is assigned to the same policies as the benchmark.}
\label{fig:real-premium-quantiles}
\end{figure}

\subsection{Marginal differentiation and policy-level allocation}
\label{subsec:real-main-results}

The endpoints and the canonical path behave as required. The pooled tariff has both coordinates equal to zero, the benchmark has both coordinates equal to one, and benchmark shrinkage returns $D_1^{\mathrm{cpl}}=D_1^W=0.600$. These results provide a direct numerical check of the common normalization.

The fitted tariffs reveal a sharp difference between marginal differentiation and policy-level allocation. Relative to the XGBoost benchmark, the two GLM tariffs have coupled coordinates of $0.320$ and $0.324$ but marginal coordinates of $0.891$ and $0.894$. Their normalized allocation mismatch is therefore about $0.57$. The GLMs generate a premium distribution that is relatively close to the benchmark marginally, while much of that dispersion is attached to different policies.

The XGBoost tariffs display the opposite pattern. Their marginal coordinates, $0.806$ and $0.811$, are somewhat lower than those of the GLMs, but their coupled coordinates rise to $0.655$ without $S$ and $0.680$ with $S$. Their normalized mismatch falls to $0.151$ and $0.131$. Thus the boosted-tree rules are substantially better aligned with benchmark risk at policy level, even though their marginal premium distributions are not the closest in Wasserstein distance. This contrast is the empirical reason to report the two coordinates jointly: premium dispersion alone does not identify who receives the high and low premiums.

Conditional multiplier-bootstrap intervals reported in Appendix~\ref{app:bootstrap-uncertainty} are narrow relative to these contrasts. They confirm that the difference between the GLM tariffs---high marginal profiling but weak policy-level alignment---and the XGBoost tariffs---stronger individual alignment and lower mismatch---is not an artifact of finite-portfolio composition.

Figure~\ref{fig:real-premium-quantiles} shows the marginal component of this result. The pooled rule is constant and the benchmark has the widest tails. The shrinkage rule follows the same ordering while compressing the distribution toward its mean. The fitted GLM and XGBoost curves lie between these endpoints. The figure cannot, however, display the allocation mismatch that separates the GLM and XGBoost results in Table~\ref{tab:cas-metrics-compact}.

Direct use of \texttt{Gender} changes the coupled profiling coordinate only modestly. Adding $S$ increases $D_1^{\mathrm{cpl}}$ from $0.320$ to $0.324$ for the GLM and from $0.655$ to $0.680$ for XGBoost. This does not imply that gender is irrelevant to group disparities. It means only that, conditional on the other variables and on the chosen benchmark, its incremental contribution to policy-level profiling is small. The group diagnostics tell a different story: the premium gap and the conditional $W_2$ gap can remain sizeable even when the variable itself adds little incremental predictive differentiation.

\subsection{Group parity and benchmark-relative transfers}
\label{subsec:real-group-parity}

For each original tariff, we apply the one-dimensional barycentric group-parity correction of Section~\ref{subsec:group-parity-mutualization}. Let $Q_s$ be the exposure-weighted quantile function of the premium values in group $S=s$ and $p_s=\sum_{i:S_i=s}w_i$ the corresponding exposure weight. The barycentric quantile is
\[
Q_B(u)=\sum_s p_sQ_s(u).
\]
Each premium is mapped to $Q_B(u_i)$, where $u_i$ is its exposure-weighted within-group midrank, and the resulting tariff is rebalanced to the common portfolio mean. Because the empirical conditional distributions are discrete and have unequal weights, the deterministic midrank implementation produces approximate rather than exact equality of the two conditional distributions. The remaining $W_2$ gaps should therefore be read as numerical residuals of the empirical mapping, not as a failure of the population construction.

\begin{table}[t]
\centering
\small
\caption{Effect of Wasserstein fairness post-processing. Entries are fair minus original. Negative changes in the two group-gap columns indicate reduced gender disparities; positive changes in the transfer gap indicate compensating benchmark-relative transfers.}
\label{tab:cas-fairness-effects}
\begin{tabular}{lrrrrrr}
\toprule
Tariff & $\Delta D_1^{\mathrm{cpl}}$ & $\Delta D_1^W$ & $\Delta\overline M_1$ & $\Delta W_2$ gap & $\Delta$ premium gap & $\Delta$ transfer gap \\
\midrule
GLM, no $S$         & -0.006 &  0.002 & 0.009 &  -8.631 &  -4.590 & 4.591 \\
GLM, with $S$       & -0.008 &  0.002 & 0.010 & -12.331 & -10.696 & 8.241 \\
XGBoost, no $S$     & -0.004 &  0.000 & 0.004 &  -6.025 &  -1.966 & 1.971 \\
XGBoost, with $S$   & -0.014 & -0.003 & 0.011 &  -9.830 &  -7.670 & 7.674 \\
Benchmark shrinkage & -0.012 & -0.003 & 0.008 &  -7.185 &  -5.674 & 5.688 \\
Benchmark           & -0.117 & -0.009 & 0.108 & -11.974 &  -9.456 & 9.481 \\
\bottomrule
\end{tabular}
\end{table}

The post-processing sharply reduces group premium disparities. Table~\ref{tab:cas-fairness-effects-compact} reports large negative changes in both the conditional $W_2$ gap and the mean premium gap for the tariffs fitted with $S$, and similar reductions occur for the rules fitted without $S$. This confirms that unawareness is not group neutrality: vehicle and driver variables preserve part of the gender-related structure even when the sensitive attribute is excluded from the model.

The pooling--profiling coordinates react asymmetrically. At the benchmark endpoint, the parity-corrected transformation produces a visible reduction in $D_1^{\mathrm{cpl}}$, while $D_1^W$ remains close to one. The normalized mismatch therefore increases: the common marginal distribution remains close to the benchmark distribution, but part of that dispersion is reassigned across policies in order to reduce the group gap. The same mechanism appears, on a smaller scale, for the fitted tariffs. Across those tariffs, the change in the coupled coordinate is systematically more pronounced than the change in the marginal coordinate; Table~\ref{tab:cas-fairness-effects-compact} reports the corresponding differences.

The transfer columns give the corresponding redistributive interpretation. Before correction, the benchmark tariff reproduces the benchmark gender premium gap of about $9.554$ and creates essentially no benchmark-relative group transfer. After post-processing, the premium gap is nearly eliminated while the transfer gap increases correspondingly. For \texttt{xgb\_S}, the original premium gap is about $7.487$ and the original transfer gap about $2.067$; the parity correction again shifts much of the remaining group difference from the premium schedule into benchmark-relative transfers. Conditional distributional parity therefore does not eliminate the underlying benchmark difference. It determines whether that difference appears in premiums or in transfers.

\begin{figure}[t]
\centering
\includegraphics[width=.86\textwidth]{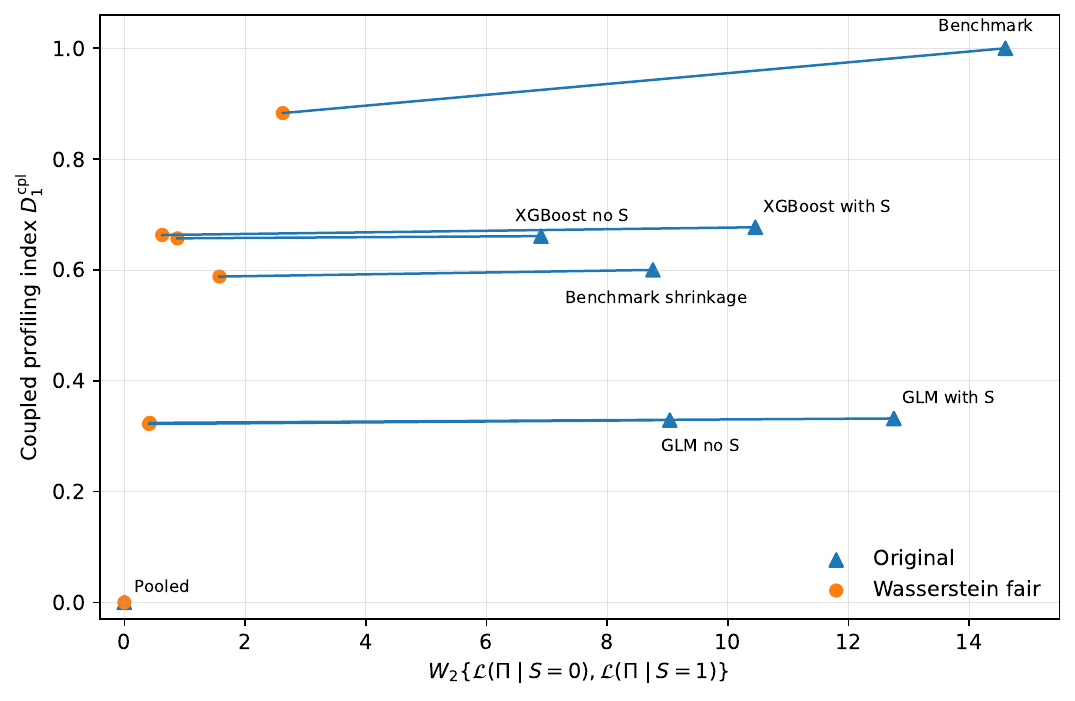}
\caption{Group parity and pricing demutualization. Each segment connects an original tariff to its barycentric parity-corrected version. Moving left indicates a reduction in the gender-conditional premium-distribution gap; moving down indicates a reduction in coupled demutualization and a corresponding increase in residual mutualization.}
\label{fig:real-tradeoff}
\end{figure}

Figure~\ref{fig:real-tradeoff} summarizes the adjustment in two dimensions. Every non-pooled tariff moves sharply to the left. The vertical movement is limited for already imperfect candidate tariffs and largest at the benchmark endpoint, where all group-level benchmark differences initially enter the premium. The diagram distinguishes a reduction in group premium disparity from a reduction in monetary differentiation: the former may be achieved mainly through reallocation rather than through compression of the marginal premium distribution.

\subsection{Dependence on the profiling benchmark}
\label{subsec:real-benchmark-sensitivity}

The numerical coordinates are conditional on the accepted profiling frontier. To assess this dependence, we recompute the coordinates using three cross-fitted benchmarks based on the same covariates: a GLM, a GAM and the more flexible XGBoost reference. Table~\ref{tab:benchmark-sensitivity-compact} reports the results for the four candidate tariffs.

\begin{table}[t]
\centering
\small
\caption{Sensitivity of the pooling--profiling coordinates to the empirical profiling benchmark. Only candidate tariffs are shown.}
\label{tab:benchmark-sensitivity-compact}
\begin{tabular}{llrrr}
\toprule
Benchmark & Tariff & $D_1^{\mathrm{cpl}}$ & $D_1^W$ & $\overline M_1$ \\
\midrule
GLM       & GLM, no $S$       & 0.928 & 0.988 & 0.060 \\
GLM       & GLM, with $S$     & 1.000 & 1.000 & 0.000 \\
GLM       & XGBoost, no $S$   & 0.424 & 0.848 & 0.424 \\
GLM       & XGBoost, with $S$ & 0.429 & 0.852 & 0.422 \\
\addlinespace
GAM       & GLM, no $S$       & 0.745 & 0.945 & 0.200 \\
GAM       & GLM, with $S$     & 0.756 & 0.947 & 0.191 \\
GAM       & XGBoost, no $S$   & 0.504 & 0.821 & 0.317 \\
GAM       & XGBoost, with $S$ & 0.512 & 0.823 & 0.311 \\
\addlinespace
XGBoost reference & GLM, no $S$       & 0.329 & 0.887 & 0.559 \\
XGBoost reference & GLM, with $S$     & 0.332 & 0.891 & 0.558 \\
XGBoost reference & XGBoost, no $S$   & 0.661 & 0.808 & 0.147 \\
XGBoost reference & XGBoost, with $S$ & 0.677 & 0.809 & 0.132 \\
\bottomrule
\end{tabular}
\end{table}

The levels change materially across benchmarks. Relative to the GLM endpoint, the two GLM tariffs have coupled coordinates of $0.928$ and $1.000$, whereas the XGBoost tariffs are around $0.42$. Relative to the GAM endpoint, the GLM coordinates are $0.563$ and $0.573$, while the XGBoost coordinates are $0.421$ and $0.429$. Relative to the XGBoost reference, the ordering reverses: the GLMs fall to $0.320$--$0.324$, while the XGBoost tariffs reach $0.655$--$0.680$. In this implementation, the value $D_1^{\mathrm{cpl}}=D_1^W=1$ for the GLM with $S$ under the GLM benchmark is an endpoint identity by construction: the benchmark is generated by the same cross-fitted GLM specification as that candidate tariff. It should therefore be read as a normalization check rather than as an independent performance result.

This variation is the empirical content of benchmark relativity, not a failure of normalization. A tariff is not intrinsically ``$67\%$ profiled'' independently of the reference model. It is $67\%$ profiled relative to a stated covariate set and modelling protocol. The benchmark ladder therefore serves two purposes. First, it identifies conclusions that are stable, such as the large difference between marginal and coupled profiling for the GLMs. Second, it reveals conclusions that depend on the accepted representation of individual risk, such as the relative proximity of GLM and XGBoost tariffs to the profiling endpoint.

The results are essentially insensitive to the balancing convention. Replacing multiplicative balance by additive balance changes the coupled coordinate by less than $0.0006$ and the marginal Wasserstein coordinate by less than $0.0012$ across the four fitted tariffs. No non-positive premium is produced: the minimum additively balanced premium is approximately $47.6$. Appendix~\ref{app:balancing-sensitivity} reports the full comparison.

Additional residual ordered-Lorenz, local-slope and variable-removal diagnostics are reported in the appendices. They support the same qualitative interpretation without changing the three principal findings above. First, marginal premium differentiation and individual benchmark alignment can diverge sharply. Second, excluding a sensitive attribute does not restore group neutrality or pooling when correlated covariates remain. Third, barycentric group parity converts much of the group premium difference into benchmark-relative transfers while leaving marginal premium differentiation almost unchanged.

\section{Discussion and conclusion}
\label{sec:conclusion}

This paper proposes a benchmark-relative pooling--profiling scale for insurance pricing. The coupled coordinate $D_p^{\rm cpl}$ measures policy-level alignment with a stated benchmark, while the marginal Wasserstein coordinate $D_p^W$ measures whether the tariff reproduces the benchmark's monetary distribution. The inequality $D_p^W\geq D_p^{\rm cpl}$ formalizes the distinction between marginal differentiation and policy-level allocation, while $M_p$ measures the associated residual-cost mismatch. For $p=1$, the coupled coordinate has an exact transfer interpretation: under balance, it is the fraction of the benchmark-relative transfer volume induced by full pooling that the evaluated tariff removes.

The applications show why the distinction matters. Broad classes and GLM tariffs can reproduce much of the benchmark dispersion while retaining substantial allocation mismatch. Tail caps affect $p=2$ more strongly than $p=1$, while shrinkage along the canonical path compresses differentiation throughout the distribution. In the motor portfolio, barycentric group parity sharply reduces gender-conditional premium disparities mainly through reallocation: coupled profiling falls and benchmark-relative transfers increase, whereas the marginal Wasserstein coordinate changes little.

The coordinates are conditional on the benchmark covariates, learner, out-of-sample protocol and balancing convention. They are fixed-portfolio quantities and do not incorporate participation, coverage or competitive responses. They also do not identify an optimal degree of mutualization or determine which rating factors are legitimate. Their purpose is narrower: to make observable how much benchmark heterogeneity is translated into prices, how much remains pooled, and whether the resulting monetary differentiation is assigned to the same policies as under the stated benchmark.



\section*{Author Contributions}
Conceptualization: 
A.C. and L.B.; 
methodology: A.C. and L.B.;
formal analysis: A.C.; 
software: A.C.; 
data curation: A.C.;
visualization: A.C.; 
writing---original draft: A.C. and L.B.;
writing---review and editing: A.C. and L.B.
Both authors approved the final version of the manuscript.

\section*{Declaration of generative AI use}

Generative AI tools were used solely for English-language polishing, as neither author is a native English speaker, and for minor cleaning of the \LaTeX{} source code. All scientific content was written, checked, and approved by the authors, who take full responsibility for the manuscript.

\section*{Competing interests}

The authors declare no competing interests.

\section*{Data availability statement}

The data used in the empirical application are publicly available
through the \texttt{CASdatasets} R package. The synthetic data are
generated by the replication code. The code required to reproduce the
simulations, tables, and figures will be made publicly available in a
GitHub repository upon publication, \href{https://github.com/freakonometrics/demutualization/}{https://github.com/freakonometrics/demutualization/}.

\clearpage
\appendix

\section{General homogeneous discrepancies}
\label{sup:general-discrepancies}

The main text defines directly the two working coordinates used in the applications. This section records the more general normalization that motivated them and clarifies which properties require additional structure.

\subsection{Coupled discrepancies}

Let $g_{\bw}:\mathbb R^n\to[0,\infty)$ be continuous, positive definite, invariant under joint relabelling, and positively homogeneous of order one:
\[
g_{\bw}(c\bz)=|c|g_{\bw}(\bz).
\]
The discrepancy
\[
d_{\bw}(\bx,\by)=g_{\bw}(\bx-\by)
\]
induces
\[
D_d^{\rm cpl}(\bpi;\btheta,\bw)
=
1-\frac{d_{\bw}(\bpi,\btheta)}{d_{\bw}(\bar\theta\boldsymbol 1,\btheta)}.
\]
Positive homogeneity is sufficient for exact calibration on the canonical pooling--profiling path. Subadditivity is not needed for that normalization; it is needed only if one wants $d_{\bw}$ to be a metric. The weighted $L^p$ norm used in the main text satisfies all of these properties.

\subsection{Marginal discrepancies}

For a discrepancy $\rho$ between probability distributions on $\mathbb R$, the analogous construction is
\[
D_\rho^{\rm marg}(\bpi;\btheta,\bw)
=
1-\frac{\rho(\mu_{\bpi},\mu_{\btheta})}{\rho(\delta_{\bar\theta},\mu_{\btheta})}.
\]
Endpoint normalization requires that $\rho(\mu,\nu)=0$ if and only if $\mu=\nu$. Monetary scale invariance follows if, for $c>0$,
\[
\rho\{(x\mapsto cx)_\#\mu,(x\mapsto cx)_\#\nu\}=c\rho(\mu,\nu).
\]
Exact canonical-path calibration additionally requires, for the affine contraction $T_a(x)=(1-a)m+ax$ around the benchmark mean $m$, that
\[
\rho\{(T_a)_\#\mu,\mu\}=(1-a)\rho(\delta_m,\mu),
\qquad 0\leq a\leq1.
\]
Wasserstein distances satisfy these properties and are therefore used throughout the main text. An arbitrary distributional discrepancy need not. For example, total variation is invariant to many changes in support location and generally does not decrease linearly along the canonical path; its normalized version therefore fails the desired calibration. This is why the revised main text no longer treats every distributional discrepancy as automatically defining a pooling--profiling coordinate.

\section{Population limits of the empirical coordinates}
\label{app:population-limits}

The preceding indices are finite-portfolio quantities. They nevertheless have a
natural population interpretation when the portfolio is viewed as an empirical
sample from an underlying distribution. The following result makes this point in
the simplest i.i.d.\ setting. It also clarifies why the coupled and distributional
indices should remain distinct: their empirical limits depend on different
population objects.

\begin{proposition}[Almost sure population limits]
\label{prop:as-population-limits}
Let $(\Pi_i,\Theta_i)_{i\geq1}$ be i.i.d.\ copies of a pair
$(\Pi,\Theta)$ on $\mathbb R^2$. Assume that, for some $p\geq1$,
\[
\mathbb E\{|\Pi|^p+|\Theta|^p\}<\infty,
\qquad
\mathbb E|\Theta-\mathbb E[\Theta]|^p>0.
\]
For the portfolio of size $n$, take equal weights $w_{i,n}=1/n$ and set
\[
\mu_{\bpi,n}=\frac1n\sum_{i=1}^n\delta_{\Pi_i},
\qquad
\mu_{\btheta,n}=\frac1n\sum_{i=1}^n\delta_{\Theta_i},
\qquad
\bar\Theta_n=\frac1n\sum_{i=1}^n\Theta_i.
\]
Then the marginal Wasserstein coordinate
\[
D_{p,n}^{W}
=
1-
\frac{W_p(\mu_{\bpi,n},\mu_{\btheta,n})}
{W_p(\delta_{\bar\Theta_n},\mu_{\btheta,n})}
\]
converges almost surely to
\[
D_p^{W,\infty}
=
1-
\frac{W_p(\mathcal L(\Pi),\mathcal L(\Theta))}
{\{\mathbb E|\Theta-\mathbb E[\Theta]|^p\}^{1/p}}.
\]
Moreover, the coupled $L^p$ coordinate
\[
D_{p,n}^{\mathrm{cpl}}
=
1-
\frac{\left(n^{-1}\sum_{i=1}^n|\Pi_i-\Theta_i|^p\right)^{1/p}}
{\left(n^{-1}\sum_{i=1}^n|\Theta_i-\bar\Theta_n|^p\right)^{1/p}}
\]
converges almost surely to
\[
D_p^{\mathrm{cpl},\infty}
=
1-
\frac{\{\mathbb E|\Pi-\Theta|^p\}^{1/p}}
{\{\mathbb E|\Theta-\mathbb E[\Theta]|^p\}^{1/p}}.
\]
Thus the marginal Wasserstein coordinate has a population limit determined only by the
marginal laws of $\Pi$ and $\Theta$, whereas the coupled coordinate has a population
limit determined by their joint law.
\end{proposition}

\begin{proof}
The moment assumption implies, by the strong law of large numbers, that
$\bar\Theta_n\to\mathbb E[\Theta]$ almost surely and
\[
\frac1n\sum_{i=1}^n|\Pi_i-\Theta_i|^p
\longrightarrow
\mathbb E|\Pi-\Theta|^p
\qquad\text{a.s.}
\]
It also implies convergence of the empirical marginal measures in Wasserstein
distance of order $p$,
\[
W_p(\mu_{\bpi,n},\mathcal L(\Pi))\to0,
\qquad
W_p(\mu_{\btheta,n},\mathcal L(\Theta))\to0,
\qquad\text{a.s.}
\]
Hence, by the triangle inequality,
\[
W_p(\mu_{\bpi,n},\mu_{\btheta,n})
\longrightarrow
W_p(\mathcal L(\Pi),\mathcal L(\Theta))
\qquad\text{a.s.}
\]
Furthermore,
\[
W_p^p(\delta_{\bar\Theta_n},\mu_{\btheta,n})
=
\frac1n\sum_{i=1}^n|\Theta_i-\bar\Theta_n|^p.
\]
Since $\bar\Theta_n\to\mathbb E[\Theta]$ almost surely and
$n^{-1}\sum_i|\Theta_i-\mathbb E[\Theta]|^p\to
\mathbb E|\Theta-\mathbb E[\Theta]|^p$ almost surely, the reverse triangle
inequality for empirical $L^p$ norms gives
\[
\left|
\left(\frac1n\sum_{i=1}^n|\Theta_i-\bar\Theta_n|^p\right)^{1/p}
-
\left(\frac1n\sum_{i=1}^n|\Theta_i-\mathbb E[\Theta]|^p\right)^{1/p}
\right|
\leq
|\bar\Theta_n-\mathbb E[\Theta]|,
\]
and therefore
\[
W_p(\delta_{\bar\Theta_n},\mu_{\btheta,n})
\longrightarrow
\{\mathbb E|\Theta-\mathbb E[\Theta]|^p\}^{1/p}
\qquad\text{a.s.}
\]
The assumed non-degeneracy of $\Theta$ in $L^p$ keeps the limiting denominator
strictly positive. Taking ratios yields the convergence of $D_{p,n}^{W}$ and
$D_{p,n}^{\mathrm{cpl}}$.
\end{proof}

\begin{remark}[Unequal weights]
\label{rem:unequal-weights}
The equal-weight assumption is used only to keep the statement transparent. The
same proof applies to deterministic triangular arrays of portfolio weights
$w_{i,n}$ whenever the corresponding weighted empirical marginals converge
almost surely in $W_p$ and the relevant weighted $p$th-moment averages satisfy
a strong law. We do not pursue primitive sufficient conditions here.
\end{remark}

Proposition~\ref{prop:as-population-limits} provides a population reading of the finite-portfolio coordinates without making the coupled and marginal notions collapse into one another (see also the plug-in result below for a convergence in probability). In large portfolios, the marginal Wasserstein coordinate measures a marginal discrepancy between premium and benchmark distributions, while the coupled coordinate measures a joint discrepancy between assigned premiums and
benchmark risks.

\section{Finite-portfolio benchmark stability}
\label{sup:benchmark-stability}

\begin{proposition}[Plug-in stability]
\label{sup:prop:benchmark-stability}
Let $\btheta_0=(\theta_{01},\ldots,\theta_{0n})$ be a benchmark vector, and let $\widehat{\btheta}=(\widehat\theta_1,\ldots,\widehat\theta_n)$ be an estimated benchmark. Suppose that
\[
\|\widehat{\btheta}-\btheta_0\|_{p,\bw}\leq\varepsilon
\]
and that the benchmark dispersion is bounded away from zero:
\[
\|\btheta_0-\bar\theta_0\boldsymbol 1\|_{p,\bw}\geq b_p>0,
\qquad
\|\widehat{\btheta}-\bar{\widehat\theta}\boldsymbol 1\|_{p,\bw}\geq b_p/2.
\]
For a fixed evaluated premium vector $\bpi$, $D_p^{\mathrm{cpl}}(\bpi;\widehat{\btheta})$ differs from $D_p^{\mathrm{cpl}}(\bpi;\btheta_0)$ by $O(\varepsilon)$. Moreover,
\[
W_p(\mu_{\widehat{\btheta}},\mu_{\btheta_0})\leq \varepsilon,
\]
so $D_p^W(\bpi;\widehat{\btheta})$ also differs from $D_p^W(\bpi;\btheta_0)$ by $O(\varepsilon)$, provided the denominator is bounded away from zero.
\end{proposition}

\subsection{Proof}

We first prove the coupled part. The numerator is Lipschitz:
\[
\left|
\|\bpi-\widehat{\btheta}\|_{p,\bw}
-
\|\bpi-\btheta_0\|_{p,\bw}
\right|
\leq
\|\widehat{\btheta}-\btheta_0\|_{p,\bw}
\leq
\varepsilon.
\]
For the denominator,
\[
\left|
\|\widehat{\btheta}-\bar{\widehat\theta}\boldsymbol 1\|_{p,\bw}
-
\|\btheta_0-\bar\theta_0\boldsymbol 1\|_{p,\bw}
\right|
\leq
\|(\widehat{\btheta}-\btheta_0)-(\bar{\widehat\theta}-\bar\theta_0)\boldsymbol 1\|_{p,\bw}.
\]
Since
\[
|\bar{\widehat\theta}-\bar\theta_0|
=
\left|\sum_iw_i(\widehat\theta_i-\theta_{0i})\right|
\leq
\|\widehat{\btheta}-\btheta_0\|_{p,\bw}
\]
by Jensen's inequality, the right-hand side is bounded by $2\varepsilon$. Combining these numerator and denominator bounds with
\[
\|\btheta_0-\bar\theta_0\boldsymbol 1\|_{p,\bw}\geq b_p,
\qquad
\|\widehat{\btheta}-\bar{\widehat\theta}\boldsymbol 1\|_{p,\bw}\geq b_p/2,
\]
gives
\[
\left|
D_p^{\mathrm{cpl}}(\bpi;\widehat{\btheta})
-
D_p^{\mathrm{cpl}}(\bpi;\btheta_0)
\right|
\leq
\frac{2\varepsilon}{b_p}
+
\frac{4\|\bpi-\btheta_0\|_{p,\bw}}{b_p^2}\varepsilon,
\]
up to the same first-order constants as in the main text. In particular, the difference is $O(\varepsilon)$.

For the Wasserstein part, consider the empirical coupling that pairs $\widehat\theta_i$ with $\theta_{0i}$ and assigns mass $w_i$. This coupling is admissible between $\mu_{\widehat{\btheta}}$ and $\mu_{\btheta_0}$. Hence
\[
W_p^p(\mu_{\widehat{\btheta}},\mu_{\btheta_0})
\leq
\sum_iw_i|\widehat\theta_i-\theta_{0i}|^p
=
\|\widehat{\btheta}-\btheta_0\|_{p,\bw}^p,
\]
which gives
\[
W_p(\mu_{\widehat{\btheta}},\mu_{\btheta_0})
\leq
\varepsilon.
\]
The stability of $D_p^W$ then follows from the triangle inequality and the same denominator argument.

Suppose now that a fixed raw premium vector $\bpi^{\rm raw}$, with $\bar\pi^{\rm raw}>0$, is rebalanced separately to the two benchmarks:
\[
\bpi_0
=
\bpi^{\rm raw}\frac{\bar\theta_0}{\bar\pi^{\rm raw}},
\qquad
\widehat\bpi
=
\bpi^{\rm raw}\frac{\bar{\widehat\theta}}{\bar\pi^{\rm raw}}.
\]
Since $|\bar{\widehat\theta}-\bar\theta_0|\leq\varepsilon$,
\[
\|\widehat\bpi-\bpi_0\|_{p,\bw}
\leq
\frac{\|\bpi^{\rm raw}\|_{p,\bw}}{\bar\pi^{\rm raw}}\,\varepsilon.
\]
Adding this term to the coupled numerator perturbation above proves the same $O(\varepsilon)$ stability for the empirically rebalanced premiums. For the marginal coordinate, the empirical coupling of $\widehat\pi_i$ and $\pi_{0i}$ gives the same bound in $W_p$. Thus both conclusions remain valid, provided the relevant means and benchmark dispersions stay bounded away from zero.

\section{Plug-in benchmark estimates}
\label{app:estimated-benchmark}

The main text treats the benchmark vector $\btheta$ as given. In empirical
applications, however, $\btheta$ is often obtained from a fitted model, an
oracle approximation, or a regulatory reference model. This appendix records a
simple plug-in consistency result. It shows that replacing the benchmark by a
weakly consistent estimate changes the convergence mode from almost sure
convergence to convergence in probability, unless the benchmark estimator is
itself strongly consistent.

\begin{proposition}[Plug-in consistency for an estimated benchmark]
\label{prop:plugin-benchmark}
Let $(\Pi_i,\Theta_i)_{i\geq1}$ be i.i.d.\ copies of $(\Pi,\Theta)$ satisfying
\[
\mathbb E\{|\Pi|^p+|\Theta|^p\}<\infty,
\qquad
\mathbb E|\Theta-\mathbb E[\Theta]|^p>0,
\]
for some $p\geq1$. Let $\widehat\Theta_{i,n}$ be benchmark estimates and define
\[
\widehat\mu_{\btheta,n}
=
\frac1n\sum_{i=1}^n\delta_{\widehat\Theta_{i,n}},
\qquad
\widehat{\bar\Theta}_n
=
\frac1n\sum_{i=1}^n\widehat\Theta_{i,n}.
\]
Assume the empirical $L^p$ benchmark error satisfies
\[
\left(\frac1n\sum_{i=1}^n
|\widehat\Theta_{i,n}-\Theta_i|^p\right)^{1/p}
\xrightarrow{\mathbb P}0.
\]
Then the plug-in marginal Wasserstein coordinate
\[
\widehat D_{p,n}^{W}
=
1-
\frac{W_p(\mu_{\bpi,n},\widehat\mu_{\btheta,n})}
{W_p(\delta_{\widehat{\bar\Theta}_n},\widehat\mu_{\btheta,n})}
\]
converges in probability to
\[
D_p^{W,\infty}
=
1-
\frac{W_p(\mathcal L(\Pi),\mathcal L(\Theta))}
{\{\mathbb E|\Theta-\mathbb E[\Theta]|^p\}^{1/p}}.
\]
Similarly, the plug-in coupled coordinate
\[
\widehat D_{p,n}^{\mathrm{cpl}}
=
1-
\frac{\left(n^{-1}\sum_{i=1}^n|\Pi_i-\widehat\Theta_{i,n}|^p\right)^{1/p}}
{\left(n^{-1}\sum_{i=1}^n|\widehat\Theta_{i,n}-\widehat{\bar\Theta}_n|^p\right)^{1/p}}
\]
converges in probability to
\[
D_p^{\mathrm{cpl},\infty}
=
1-
\frac{\{\mathbb E|\Pi-\Theta|^p\}^{1/p}}
{\{\mathbb E|\Theta-\mathbb E[\Theta]|^p\}^{1/p}}.
\]
If the empirical benchmark error converges almost surely instead of in
probability, then both plug-in convergences above also hold almost surely.
\end{proposition}

\begin{proof}
Let
\[
\varepsilon_n
=
\left(\frac1n\sum_{i=1}^n
|\widehat\Theta_{i,n}-\Theta_i|^p\right)^{1/p}.
\]
By assumption, $\varepsilon_n\xrightarrow{\mathbb P}0$. The empirical coupling
that pairs each $\widehat\Theta_{i,n}$ with $\Theta_i$ gives
\[
W_p(\widehat\mu_{\btheta,n},\mu_{\btheta,n})
\leq
\varepsilon_n,
\]
so $\widehat\mu_{\btheta,n}$ and $\mu_{\btheta,n}$ are asymptotically equivalent
in $W_p$ in probability. Also,
\[
|\widehat{\bar\Theta}_n-\bar\Theta_n|
\leq
\frac1n\sum_{i=1}^n |\widehat\Theta_{i,n}-\Theta_i|
\leq
\varepsilon_n,
\]
where the last inequality follows from Jensen's inequality. Hence
$\widehat{\bar\Theta}_n-\bar\Theta_n\xrightarrow{\mathbb P}0$.

For the numerator of the marginal Wasserstein coordinate, the triangle inequality yields
\[
\left|
W_p(\mu_{\bpi,n},\widehat\mu_{\btheta,n})
-
W_p(\mu_{\bpi,n},\mu_{\btheta,n})
\right|
\leq
W_p(\widehat\mu_{\btheta,n},\mu_{\btheta,n})
\leq
\varepsilon_n,
\]
which is $o_{\mathbb P}(1)$. For the denominator, using
$W_p^p(\delta_a,\nu)=\int |x-a|^p\,d\nu(x)$ and the reverse triangle inequality
for empirical $L^p$ norms gives
\[
\left|
W_p(\delta_{\widehat{\bar\Theta}_n},\widehat\mu_{\btheta,n})
-
W_p(\delta_{\bar\Theta_n},\mu_{\btheta,n})
\right|
\leq
2\varepsilon_n.
\]
Thus the plug-in marginal Wasserstein coordinate is asymptotically equivalent in
probability to the oracle empirical coordinate in the preceding population-limit proposition.

For the coupled numerator,
\[
\left|
\left(\frac1n\sum_{i=1}^n|\Pi_i-\widehat\Theta_{i,n}|^p\right)^{1/p}
-
\left(\frac1n\sum_{i=1}^n|\Pi_i-\Theta_i|^p\right)^{1/p}
\right|
\leq
\varepsilon_n,
\]
again by the reverse triangle inequality. The same argument applied to the
centered benchmark vectors gives
\[
\left|
\left(\frac1n\sum_{i=1}^n
|\widehat\Theta_{i,n}-\widehat{\bar\Theta}_n|^p\right)^{1/p}
-
\left(\frac1n\sum_{i=1}^n
|\Theta_i-\bar\Theta_n|^p\right)^{1/p}
\right|
\leq
2\varepsilon_n.
\]
Therefore the plug-in coupled coordinate is asymptotically equivalent in probability
to the oracle empirical coupled coordinate. Since the oracle empirical indices
converge almost surely by the preceding population-limit proposition, the
plug-in indices converge in probability to the same limits. If
$\varepsilon_n\to0$ almost surely, the same proof gives almost sure convergence.
\end{proof}

\begin{remark}[Estimated premiums]
\label{rem:estimated-premiums}
The same argument applies if the premium vector is itself estimated or
post-processed. If
\[
\left(\frac1n\sum_{i=1}^n
|\widehat\Pi_{i,n}-\Pi_i|^p\right)^{1/p}
\xrightarrow{\mathbb P}0,
\]
then replacing $\Pi_i$ by $\widehat\Pi_{i,n}$ in the plug-in indices leaves the
same population limits unchanged.
\end{remark}

\section{Technical and commercial premiums}
\label{app:commercial-premiums}

The indices are designed primarily for technical pure premiums. A charged premium may additionally include expenses, commissions, capital and profit loadings, taxes, payment-plan effects and market adjustments. These components can create premium heterogeneity without corresponding to expected-loss heterogeneity \citep{WernerModlin2016,OhlssonJohansson2010,WuthrichMerz2023}.

Write, schematically,
\[
\pi_i^{\mathrm{comm}}
=
\pi_i^{\mathrm{risk}}+\ell_i+k_i+o_i,
\]
where \(\ell_i\) contains expenses and commissions, \(k_i\) contains capital or profit loadings, and \(o_i\) contains market adjustments. For either working coordinate $D_p^\star$, with $\star\in\{\mathrm{cpl},W\}$, computing $D_p^\star(\bpi^{\mathrm{risk}};\btheta,\bw)$ measures technical profiling. Computing the same coordinate on a balanced commercial premium measures the combined effect of risk classification and the commercial layer. Their difference,
\[
\Delta_{D_p^\star}^{\mathrm{comm}}
=
D_p^\star(\bpi^{\mathrm{comm}};\btheta,\bw)
-
D_p^\star(\bpi^{\mathrm{risk}};\btheta,\bw),
\]
is a diagnostic of commercial re-profiling whose interpretation depends on the chosen coupled or distributional coordinate.

Common proportional or additive loadings disappear under the corresponding balancing convention. Heterogeneous loadings do not: a flat expense charge may compress relative premium differences, whereas risk-proportional expenses or demand-based adjustments may amplify them. Accordingly, empirical audits should use technical pure premiums whenever possible and analyze observed commercial prices as a separate layer. Price optimization raises a distinct governance issue because it may personalize final prices through willingness-to-pay or retention information rather than expected loss \citep{NAIC2015PriceOptimization,WernerModlin2016}.

\section{Canonical-path dispersion and Gini diagnostics}
\label{sup:canonical-diagnostics}

This section develops the Gini results used in the main text. The classical
Gini coefficient is a marginal dispersion functional: it summarizes the
spread of the premium distribution but does not record which policy receives
which premium. This places it naturally between elementary dispersion
summaries and the full marginal Wasserstein coordinate.

For a non-negative portfolio vector
$\bx=(x_1,\ldots,x_n)$ with positive exposure-weighted mean
\[
\bar x=\sum_i w_i x_i>0,
\]
define the exposure-weighted Gini coefficient by
\begin{equation}
\label{eq:sup-weighted-gini}
\mathcal G(\bx)
=
\frac{
\sum_i\sum_j w_iw_j|x_i-x_j|
}{
2\bar x
}.
\end{equation}
Equivalently, if $Q_x$ denotes the quantile function of the weighted
empirical distribution $\mu_{\bx}=\sum_iw_i\delta_{x_i}$, then
\begin{equation}
\label{eq:sup-gini-quantile}
\mathcal G(\bx)
=
\frac{1}{\bar x}
\int_0^1(2u-1)Q_x(u)\,du.
\end{equation}

For a balanced tariff $\bpi$, so that
$\sum_iw_i\pi_i=\bar\theta$, define the benchmark-relative Gini ratio
\begin{equation}
\label{eq:sup-gini-ratio}
D^{\rm G}(\bpi;\btheta,\bw)
=
\frac{\mathcal G(\bpi)}{\mathcal G(\btheta)},
\end{equation}
whenever the benchmark is non-degenerate. This quantity is not introduced
as a third pooling--profiling coordinate: unlike $D_1^W$, it compresses the
entire marginal premium distribution into a single scalar. It is nevertheless
useful because it admits the same exact calibration on the canonical path.

\subsection{Canonical-path identities}

Recall the canonical pooling--profiling path
\[
\bpi^{(a)}
=
(1-a)\bar\theta\boldsymbol 1+a\btheta,
\qquad 0\leq a\leq1.
\]

\begin{proposition}[Canonical-path Gini calibration]
\label{prop:sup-gini-canonical}
For every $p\geq1$,
\[
D_p^{\rm cpl}(\bpi^{(a)})
=
D_p^W(\bpi^{(a)})
=
a,
\qquad
M_p(\bpi^{(a)})=0.
\]
Moreover,
\begin{equation}
\label{eq:sup-sd-canonical}
\operatorname{sd}_{\bw}(\bpi^{(a)})
=
a\,\operatorname{sd}_{\bw}(\btheta),
\end{equation}
and
\begin{equation}
\label{eq:sup-gini-canonical}
\mathcal G(\bpi^{(a)})
=
a\,\mathcal G(\btheta).
\end{equation}
Consequently,
\begin{equation}
\label{eq:sup-three-calibrations}
D^{\rm G}(\bpi^{(a)})
=
D_1^W(\bpi^{(a)})
=
D_1^{\rm cpl}(\bpi^{(a)})
=
a.
\end{equation}
\end{proposition}

Thus, on the canonical path, the parameter $a$ has three simultaneous
interpretations: it is the fraction of benchmark Gini dispersion retained
in premiums, the degree of marginal Wasserstein profiling, and the degree
of policy-level coupled profiling. The equality is specific to this path.
Away from it, the three summaries need not agree.

\begin{proof}
The identities for $D_p^{\rm cpl}$, $D_p^W$ and $M_p$ follow from the
homogeneity and monotone one-dimensional transport arguments given in the
main proofs. Since
\[
\pi_i^{(a)}-\bar\theta
=
a(\theta_i-\bar\theta),
\]
equation \eqref{eq:sup-sd-canonical} follows immediately.

For the Gini coefficient,
\[
|\pi_i^{(a)}-\pi_j^{(a)}|
=
a|\theta_i-\theta_j|,
\]
while
\[
\sum_iw_i\pi_i^{(a)}=\bar\theta.
\]
Substitution into \eqref{eq:sup-weighted-gini} therefore gives
\[
\mathcal G(\bpi^{(a)})
=
a\mathcal G(\btheta),
\]
and division by $\mathcal G(\btheta)$ yields
\eqref{eq:sup-three-calibrations}.
\end{proof}

\subsection{Gini dispersion and the marginal Wasserstein coordinate}

The preceding identity is exact only on the canonical path. A more general
relation follows from the fact that the Gini coefficient is a linear
functional of the quantile function, whereas $W_1$ measures its full
$L^1$ discrepancy.

\begin{proposition}[Gini--Wasserstein control]
\label{prop:sup-gini-wasserstein}
Let $\bpi$ and $\btheta$ be non-negative, non-degenerate and balanced,
with common mean $\bar\theta>0$. Then
\begin{equation}
\label{eq:sup-gini-w1}
\left|
\mathcal G(\bpi)-\mathcal G(\btheta)
\right|
\leq
\frac{
W_1(\mu_{\bpi},\mu_{\btheta})
}{
\bar\theta
}.
\end{equation}
Equivalently,
\begin{equation}
\label{eq:sup-relative-gini-w1}
\frac{
|D^{\rm G}(\bpi;\btheta,\bw)-1|
}{
\kappa_\theta
}
\leq
1-D_1^W(\bpi;\btheta,\bw),
\end{equation}
where
\begin{equation}
\label{eq:sup-kappa-gini}
\kappa_\theta
=
\frac{
W_1(\delta_{\bar\theta},\mu_{\btheta})
}{
\bar\theta\,\mathcal G(\btheta)
}
\in[1,2].
\end{equation}
Since $D_1^W\geq D_1^{\rm cpl}$,
\begin{equation}
\label{eq:sup-gini-chain}
\frac{
|D^{\rm G}(\bpi;\btheta,\bw)-1|
}{
\kappa_\theta
}
\leq
1-D_1^W(\bpi;\btheta,\bw)
\leq
1-D_1^{\rm cpl}(\bpi;\btheta,\bw).
\end{equation}
\end{proposition}

The first inequality has a useful interpretation. Marginal Wasserstein
proximity to the benchmark forces Gini proximity: if $D_1^W$ is close to
one, the evaluated tariff must have approximately the same Gini coefficient
as the benchmark. The converse does not hold. Equal Gini coefficients impose
only one scalar restriction and do not imply equality, or even close
proximity, of the two marginal distributions.

\begin{proof}
Because $\bpi$ and $\btheta$ have the same positive mean, the quantile
representation \eqref{eq:sup-gini-quantile} gives
\[
\mathcal G(\bpi)-\mathcal G(\btheta)
=
\frac{1}{\bar\theta}
\int_0^1
(2u-1)
\{Q_{\bpi}(u)-Q_{\btheta}(u)\}
\,du.
\]
Hence, using $|2u-1|\leq1$,
\[
\begin{split}
\left|
\mathcal G(\bpi)-\mathcal G(\btheta)
\right|
&\leq
\frac{1}{\bar\theta}
\int_0^1
|Q_{\bpi}(u)-Q_{\btheta}(u)|
\,du
\\
&=
\frac{
W_1(\mu_{\bpi},\mu_{\btheta})
}{
\bar\theta
},
\end{split}
\]
which proves \eqref{eq:sup-gini-w1}.

By definition of the marginal coordinate,
\[
W_1(\mu_{\bpi},\mu_{\btheta})
=
\{1-D_1^W(\bpi;\btheta,\bw)\}
W_1(\delta_{\bar\theta},\mu_{\btheta}).
\]
Dividing \eqref{eq:sup-gini-w1} by
$\mathcal G(\btheta)$ therefore gives
\eqref{eq:sup-relative-gini-w1}.

It remains to establish the bounds on $\kappa_\theta$. Let $\Theta$ and
$\Theta'$ be independent draws from $\mu_{\btheta}$. Then
\[
W_1(\delta_{\bar\theta},\mu_{\btheta})
=
\mathbb E|\Theta-\mathbb E\Theta|,
\]
whereas
\[
\bar\theta\,\mathcal G(\btheta)
=
\frac12\mathbb E|\Theta-\Theta'|.
\]
Jensen's inequality, conditional on $\Theta$, gives
\[
\mathbb E|\Theta-\mathbb E\Theta|
\leq
\mathbb E|\Theta-\Theta'|,
\]
while the triangle inequality gives
\[
\mathbb E|\Theta-\Theta'|
\leq
2\mathbb E|\Theta-\mathbb E\Theta|.
\]
Therefore
\[
1
\leq
\frac{
2\mathbb E|\Theta-\mathbb E\Theta|
}{
\mathbb E|\Theta-\Theta'|
}
\leq2,
\]
which is exactly $\kappa_\theta\in[1,2]$.
Finally, $D_1^W\geq D_1^{\rm cpl}$ yields
\eqref{eq:sup-gini-chain}.
\end{proof}

\begin{remark}[No converse and no allocation information]
\label{rem:sup-gini-no-converse}
The Gini ratio need not lie in $[0,1]$: a balanced tariff can be more
dispersed than the benchmark and therefore have $D^{\rm G}>1$.
More importantly, $D^{\rm G}$ contains no policy-level allocation
information.

For example, consider the two-policy benchmark
\[
\btheta=(80,120)
\]
and the reversed tariff
\[
\bpi=(120,80).
\]
The two marginal distributions are identical, so
\[
D^{\rm G}=D_1^W=1.
\]
Nevertheless,
\[
D_1^{\rm cpl}=-1,
\qquad
M_1=2.
\]
Thus even equality of the \emph{entire} marginal distributions, and hence
a fortiori equality of their Gini coefficients, does not imply policy-level
alignment.
\end{remark}

\subsection{Residual ordered Gini on the canonical path}

The classical Gini coefficient above measures marginal premium dispersion.
Ordered-Lorenz diagnostics answer a different question: after sorting
policies by a residual risk relativity, how much benchmark segmentation
remains unexplained by the current tariff?

Along the canonical path, define
\[
R_i^{\rm res}
=
\frac{\theta_i}{\pi_i^{(a)}}.
\]
For $a<1$, the map
\[
\theta
\longmapsto
\frac{\theta}
{(1-a)\bar\theta+a\theta}
\]
is strictly increasing, so sorting policies by $R_i^{\rm res}$ is equivalent
to sorting them by benchmark risk $\theta_i$.

Let
\[
L_\theta(t)
=
\frac{1}{\bar\theta}
\int_0^tQ_\theta(u)\,du
\]
denote the benchmark Lorenz curve. When policies are ordered by
$R_i^{\rm res}$, the cumulative premium share under the canonical tariff is
\begin{equation}
\label{eq:sup-canonical-premium-lorenz}
P_a(t)
=
\frac{1}{\bar\theta}
\int_0^t
\{(1-a)\bar\theta+aQ_\theta(u)\}
\,du
=
(1-a)t+aL_\theta(t).
\end{equation}
The corresponding cumulative benchmark-risk share is $L_\theta(t)$.
Hence the residual ordered-Gini gap satisfies
\begin{equation}
\label{eq:sup-residual-gini}
\mathcal G_{\rm res}(\bpi^{(a)})
=
2\int_0^1
\{P_a(t)-L_\theta(t)\}
\,dt
=
(1-a)\mathcal G(\btheta).
\end{equation}
At the profiling endpoint $a=1$, the two cumulative-share curves coincide
and the residual ordered Gini is zero.

Combining \eqref{eq:sup-gini-canonical} and
\eqref{eq:sup-residual-gini} gives the decomposition
\begin{equation}
\label{eq:sup-gini-decomposition}
\mathcal G(\bpi^{(a)})
+
\mathcal G_{\rm res}(\bpi^{(a)})
=
\mathcal G(\btheta).
\end{equation}
Thus, on the canonical path, benchmark Gini dispersion splits exactly into
a priced component $a\mathcal G(\btheta)$ and a residual segmentation
component $(1-a)\mathcal G(\btheta)$.

This identity is specific to the canonical affine contraction. More
generally, classical Gini dispersion, residual ordered Gini, marginal
Wasserstein profiling and coupled profiling describe different features of
a tariff and need not admit such an additive decomposition.

\section{Barycentric group parity: proof and $L^2$ decomposition}
\label{sup:parity}

We use the notation of the main text. Let $Q_s^\Theta$ be the conditional benchmark quantiles, $Q_B^\Theta=\sum_sp_sQ_s^\Theta$, and write $c=1-a$. At rank $u$, define
\[
A(u)=Q_B^\Theta(u)-\bar\theta,
\qquad
Z_s(u)=Q_s^\Theta(u)-Q_B^\Theta(u),
\]
so that $\sum_sp_sZ_s(u)=0$. The residual benchmark error on the canonical path is $c\{A(u)+Z_s(u)\}$, while after parity correction it is $cA(u)+Z_s(u)$.

For fixed $u$ and $t>0$, set
\[
h(t)=t^{-p}\sum_sp_s|Z_s+tA|^p.
\]
With $\psi(y)=|y|^{p-1}\operatorname{sign}(y)$ (and a monotone subgradient at zero for $p=1$),
\[
h'(t)=-pt^{-p-1}\sum_sp_sZ_s\psi(Z_s+tA)\leq0,
\]
because $z\mapsto\psi(z+tA)$ is increasing and $\sum_sp_sZ_s=0$. Hence $h(c)\geq h(1)$ for $0<c\leq1$, which gives
\[
\sum_sp_s|cA+Z_s|^p
\geq
c^p\sum_sp_s|A+Z_s|^p.
\]
Integrating over $u$ proves the inequality stated in the main text.

For $p=2$, define
\[
V_B=\int_0^1\{Q_B^\Theta(u)-\bar\theta\}^2\,du,
\qquad
V_S=\sum_sp_s\int_0^1\{Q_s^\Theta(u)-Q_B^\Theta(u)\}^2\,du.
\]
Then
\[
\operatorname{Var}(\Theta)=V_B+V_S,
\]
and
\[
\|\Pi_{\rm par}^{(a)}-\Theta\|_{L^2}^2=(1-a)^2V_B+V_S,
\]
whereas
\[
\|\Pi^{(a)}-\Theta\|_{L^2}^2=(1-a)^2(V_B+V_S).
\]
Therefore
\[
\|\Pi_{\rm par}^{(a)}-\Theta\|_{L^2}^2
-
\|\Pi^{(a)}-\Theta\|_{L^2}^2
=
\{1-(1-a)^2\}V_S.
\]
The component $V_S$ is precisely the between-group Wasserstein component that must be restored to policy-level mismatch when the conditional premium distributions are equalized.

\section{Additional synthetic diagnostics}
\label{sup:synthetic}

The main text retains the premium-quantile and transfer figures. This section reports the secondary Gini and residual-segmentation diagnostics.

\begin{table}[t]
\centering
\caption{Gini diagnostics for the synthetic portfolio.}
\label{tab:synthetic-gini}
{\footnotesize
\begin{tabular}{lrrrrrr}
\toprule
Pricing rule & $D^{\mathrm{cpl}}_1$ & $D^W_1$ & $G(\pi)$ & $G(\pi)/G(\theta)$ & Residual  & Rank 
\\
 &  &  &  &  &  Gini & Gini
\\
\midrule
Pooled & 0.000 & 0.000 & 0.000 & 0.000 & 0.528 & 0.000
\\
Coarse tariff & 0.392 & 0.703 & 0.444 & 0.842 & 0.334 & 0.443
\\
GLM with proxies & 0.863 & 0.950 & 0.518 & 0.982 & 0.075 & 0.523
\\
Full GLM & 0.961 & 0.979 & 0.520 & 0.986 & 0.022 & 0.527
\\
Shrinkage & 0.590 & 0.591 & 0.312 & 0.592 & 0.218 & 0.527
\\
Capped tariff & 0.549 & 0.549 & 0.302 & 0.572 & 0.262 & 0.502
\\
No-proxy GLM & 0.145 & 0.473 & 0.280 & 0.530 & 0.459 & 0.295
\\
Oracle & 1.000 & 1.000 & 0.528 & 1.000 & 0.000 & 0.528
\\
\bottomrule
\end{tabular}}
\begin{flushleft}
\footnotesize Notes: $G(\pi)$ is the classical Gini coefficient of the premium distribution. The residual ordered Gini is computed by sorting policies according to $\theta_i/\pi_i$, following the ordered Lorenz approach of Frees, Meyers and Cummings. A high residual ordered Gini indicates that the current tariff still leaves segmentation opportunities relative to the benchmark risk score. The rank Gini measures risk ranking and is invariant to monotone transformations of the premium score.
\end{flushleft}
\end{table}

\begin{figure}[t]
\centering
\includegraphics[width=.82\textwidth]{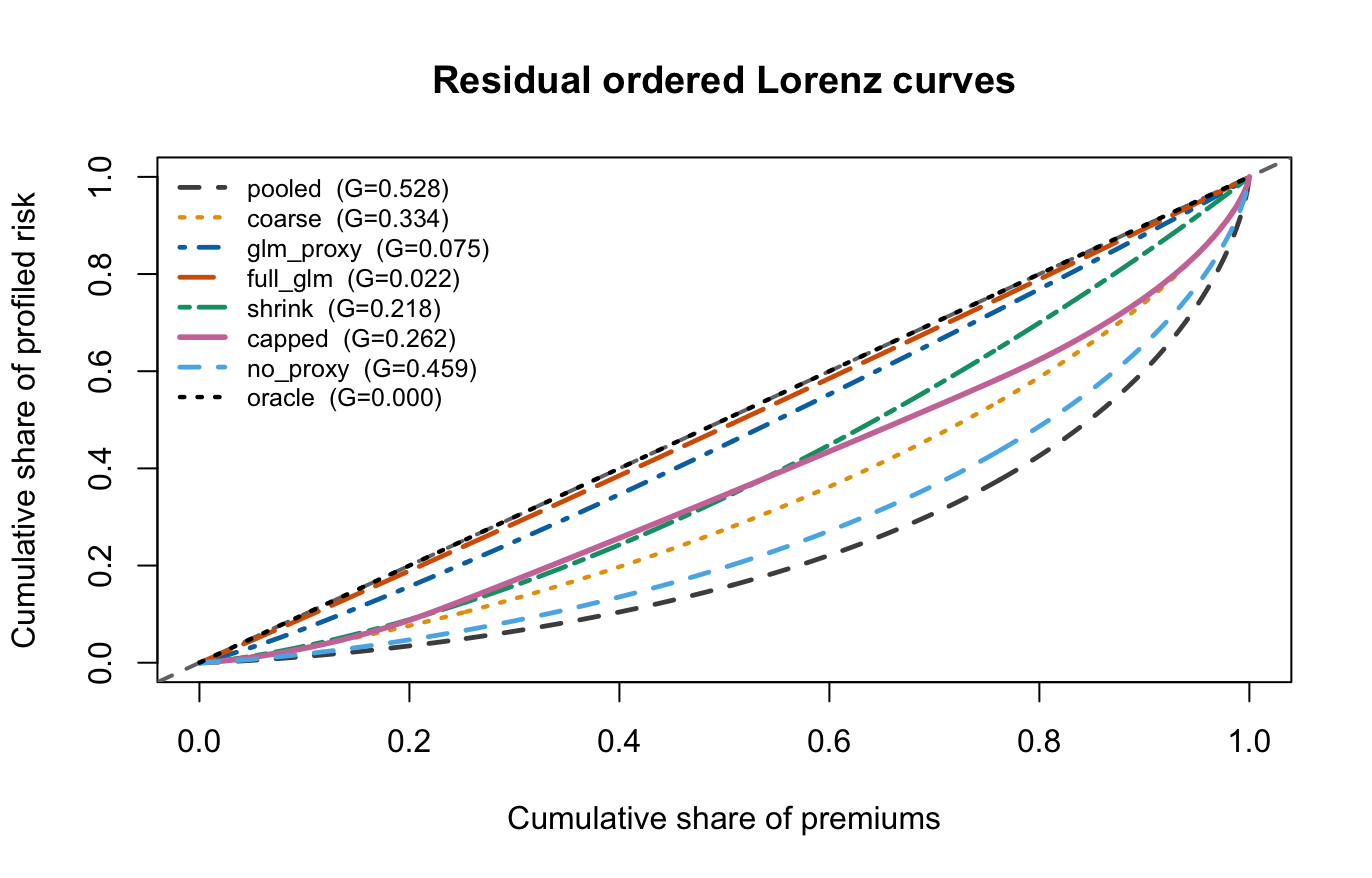}
\caption{Residual ordered Lorenz curves in the synthetic experiment, with policies sorted by $R_i^{\rm res}=\theta_i/\pi_i$. Greater curvature below the diagonal indicates more residual segmentation opportunity.}
\label{sup:fig-synthetic-lorenz}
\end{figure}

\section{Additional motor-portfolio diagnostics}
\label{sup:empirical}

Residual ordered-Lorenz curves and local slopes provide a policy-ranking diagnostic complementary to the pooling--profiling coordinates.

\begin{figure}[t]
\centering
\includegraphics[width=.82\textwidth]{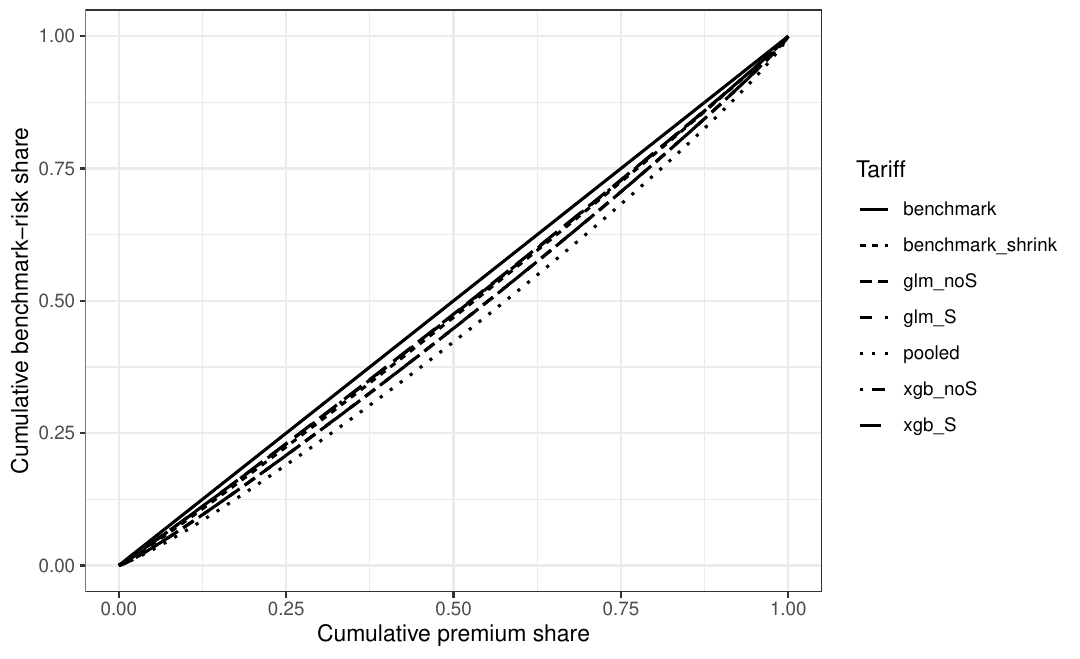}
\caption{Residual ordered Lorenz curves for the real-data tariffs. Curves close to the diagonal leave little residual segmentation opportunity relative to the stated benchmark.}
\label{sup:fig-real-residual-lorenz}
\end{figure}

\begin{figure}[t]
\centering
\includegraphics[width=.82\textwidth]{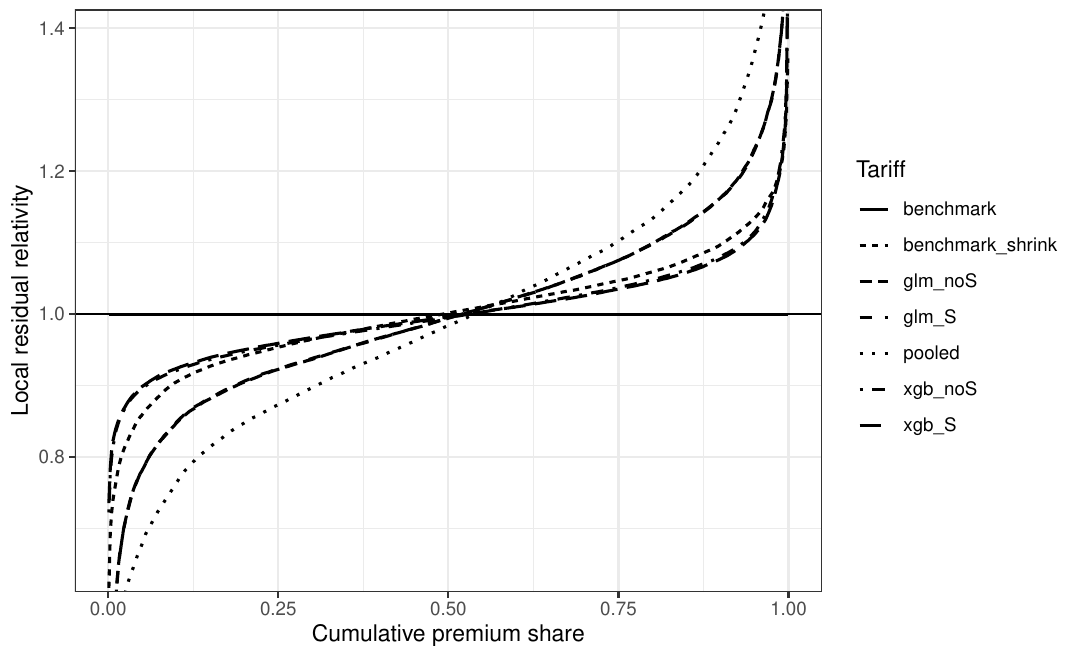}
\caption{Local slopes of the residual ordered Lorenz curves. A value of one indicates local alignment; values above one indicate benchmark risk exceeding the charged premium at that point, while values below one indicate the reverse.}
\label{sup:fig-real-residual-derivative}
\end{figure}

We also report variable-removal diagnostics for surrogate models of the XGBoost benchmark. These are predictive replacement diagnostics, not causal contributions or unique allocations of shared proxy information.

\begin{table}[t]
\centering
\small
\caption{Drop-one-variable diagnostics for surrogate models of the XGBoost benchmark. Each entry is the loss of profiling after removing the indicated covariate from the full surrogate. These are predictive diagnostics, not causal contributions.}
\label{tab:variable-drop-compact}
\begin{tabular}{lrrrr}
\toprule
& \multicolumn{2}{c}{GLM surrogate} & \multicolumn{2}{c}{XGBoost surrogate} \\
\cmidrule(lr){2-3}\cmidrule(lr){4-5}
Variable & $\Delta D_1^{\mathrm{cpl}}$ & $\Delta D_1^W$ & $\Delta D_1^{\mathrm{cpl}}$ & $\Delta D_1^W$ \\
\midrule
\texttt{VehValue} & 0.014 & 0.009 & 0.100 & 0.075 \\
\texttt{VehAge}   & 0.034 & 0.079 & 0.021 & 0.012 \\
\texttt{VehBody}  & 0.051 & 0.074 & 0.045 & 0.019 \\
\texttt{DrivAge}  & 0.220 & 0.224 & 0.233 & 0.130 \\
\texttt{Gender}   & 0.003 & 0.003 & 0.006 & 0.002 \\
\bottomrule
\end{tabular}
\end{table}

The drop-one results identify \texttt{DrivAge} as the dominant variable in both surrogate families. Removing it lowers the coupled coordinate by $0.217$ in the GLM surrogate and $0.233$ in the XGBoost surrogate. The incremental effect of \texttt{Gender} is much smaller, $0.003$ and $0.006$, respectively. This does not conflict with the group-parity results: a variable may add little unique predictive information once correlated covariates are present while still identifying groups across which premiums and benchmark-relative transfers differ.

Single-variable surrogate results are reported below for completeness.

\begin{table}[t]
\centering
\small
\caption{Single-variable surrogate profiling power. Each row reports the demutualization indices obtained by predicting the benchmark risk \(\widehat\theta\) from a single covariate.}
\label{tab:variable-single-surrogate}
\begin{tabular}{lrrrrrr}
\toprule
& \multicolumn{3}{c}{GLM surrogate} & \multicolumn{3}{c}{XGBoost surrogate} \\
\cmidrule(lr){2-4}\cmidrule(lr){5-7}
Variable & $D^{\mathrm{cpl}}_1$ & $D^W_1$ & $M_1$ & $D^{\mathrm{cpl}}_1$ & $D^W_1$ & $M_1$ \\
\midrule
\texttt{VehValue} & 0.068 & 0.322 & 11.575 & 0.153 & 0.511 & 16.327 \\
\texttt{VehAge} & 0.067 & 0.421 & 16.161 & 0.067 & 0.421 & 16.158 \\
\texttt{VehBody} & 0.043 & 0.265 & 10.126 & 0.043 & 0.253 & 9.558 \\
\texttt{DrivAge} & 0.184 & 0.495 & 14.192 & 0.183 & 0.492 & 14.103 \\
\texttt{Gender} & 0.006 & 0.095 & 4.091 & 0.006 & 0.096 & 4.101 \\
\bottomrule
\end{tabular}
\end{table}

\section{Sensitivity to the balancing convention}
\label{app:balancing-sensitivity}

The main analysis uses multiplicative balancing,
\[
\pi_i^{\mathrm{mult}}
=
\pi_i^{\mathrm{raw}}
\frac{\sum_j w_j\theta_j}
     {\sum_j w_j\pi_j^{\mathrm{raw}}},
\]
which preserves relative premium differences while enforcing equality between
weighted average premium and weighted average benchmark risk. As a robustness
check, we also consider additive balancing,
\[
\pi_i^{\mathrm{add}}
=
\pi_i^{\mathrm{raw}}
-\sum_j w_j\pi_j^{\mathrm{raw}}
+\sum_j w_j\theta_j.
\]
Both transformations impose the same portfolio mean but modify the raw tariff in
different ways.

\begin{table}[t]
\centering
\small
\caption{Sensitivity of the empirical pooling--profiling coordinates to the balancing convention. Multiplicative balancing rescales each raw tariff to the benchmark portfolio mean, whereas additive balancing shifts all premiums by a common amount. Differences are additive minus multiplicative.}
\label{tab:balancing-sensitivity}
\begin{tabular}{lrrrrrr}
\toprule
 & \multicolumn{3}{c}{$D_1^{\mathrm{cpl}}$} & \multicolumn{3}{c}{$D_1^W$} \\
\cmidrule(lr){2-4}\cmidrule(lr){5-7}
Tariff & Mult. & Add. & $\Delta$ & Mult. & Add. & $\Delta$ \\
\midrule
GLM, no gender & 0.3197 & 0.3199 & $+0.0002$ & 0.8906 & 0.8903 & $-0.0003$ \\
GLM, with gender & 0.3236 & 0.3238 & $+0.0002$ & 0.8939 & 0.8937 & $-0.0002$ \\
XGBoost, no gender & 0.6554 & 0.6558 & $+0.0004$ & 0.8063 & 0.8074 & $+0.0011$ \\
XGBoost, with gender & 0.6803 & 0.6808 & $+0.0005$ & 0.8113 & 0.8125 & $+0.0012$ \\
\bottomrule
\end{tabular}
\end{table}

Table~\ref{tab:balancing-sensitivity} shows that the resulting
pooling--profiling coordinates are virtually unchanged. Across the four fitted
tariffs, the largest absolute change is approximately $5.4\times10^{-4}$ for
$D_1^{\mathrm{cpl}}$ and $1.18\times10^{-3}$ for $D_1^W$. The additive
correction also produces no non-positive premium in this application; the
smallest additively balanced premium is approximately $47.6$. The empirical
conclusions are therefore not driven by the particular balancing convention.

\section{Conditional multiplier-bootstrap uncertainty}
\label{app:bootstrap-uncertainty}

The empirical coordinates are functionals of a weighted finite portfolio. We therefore
assess whether the main contrasts reported in the main text could be
explained by uncertainty in the composition of the observed portfolio. The exercise
is conditional on the fitted objects: the cross-fitted benchmark
$\widehat{\boldsymbol\theta}$, the original tariff vectors, and their
barycentric parity-corrected versions are held fixed.

For bootstrap replication $b$, let
$\xi_1^{(b)},\ldots,\xi_n^{(b)}$ be independent standard exponential
multipliers and define
\[
 w_i^{(b)}
 =
 \frac{w_i\xi_i^{(b)}}
 {\sum_{j=1}^n w_j\xi_j^{(b)}}.
\]
We recompute the pooling--profiling coordinates using
$\boldsymbol w^{(b)}$ in place of $\boldsymbol w$. The same multiplier weights
are used for the original and parity-corrected versions of each tariff, so uncertainty for
the parity effects is assessed from the paired bootstrap differences
\[
 \Delta D_1^{\mathrm{cpl}},
 \qquad
 \Delta D_1^W,
 \qquad
 \Delta M_1,
 \qquad
 \Delta W_2^{\mathrm{group}}.
\]
Table~\ref{tab:bootstrap-uncertainty} reports percentile $95\%$ intervals based
on $1{,}000$ replications.

\begin{table}[t]
\centering
\small
\caption{Conditional multiplier-bootstrap uncertainty. Point estimates are shown on the first line of each model entry; percentile 95\% intervals based on 1,000 bootstrap replications are reported underneath in smaller type.}
\label{tab:bootstrap-uncertainty}

\textbf{Panel A. Original tariffs}

\medskip

\begin{tabular}{lccc}
\toprule
Model
& $D_1^{\mathrm{cpl}}$
& $D_1^W$
& $\overline M_1$ \\
\midrule
GLM without $S$ & 0.329 & 0.887 & 0.559 \\
 & {\tiny [0.322, 0.335]} & {\tiny [0.880, 0.893]} & {\tiny [0.548, 0.569]} \\[-0.15em]
\addlinespace[0.35em]
GLM with $S$ & 0.332 & 0.891 & 0.558 \\
 & {\tiny [0.326, 0.338]} & {\tiny [0.883, 0.897]} & {\tiny [0.548, 0.568]} \\[-0.15em]
\addlinespace[0.35em]
XGBoost without $S$ & 0.661 & 0.808 & 0.147 \\
 & {\tiny [0.658, 0.664]} & {\tiny [0.805, 0.811]} & {\tiny [0.144, 0.150]} \\[-0.15em]
\addlinespace[0.35em]
XGBoost with $S$ & 0.677 & 0.809 & 0.132 \\
 & {\tiny [0.674, 0.680]} & {\tiny [0.806, 0.812]} & {\tiny [0.129, 0.135]} \\[-0.15em]
\addlinespace[0.35em]
\bottomrule
\end{tabular}

\bigskip

\textbf{Panel B. Fair minus original}

\medskip

\begin{tabular}{lcccc}
\toprule
Model
& $\Delta D_1^{\mathrm{cpl}}$
& $\Delta D_1^W$
& $\Delta\overline M_1$
& $\Delta W_2^{\mathrm{group}}$ \\
\midrule
GLM without $S$ & -0.0064 & +0.0022 & +0.0086 & -8.631 \\
 & {\tiny [-0.0072, -0.0056]} & {\tiny [+0.0012, +0.0033]} & {\tiny [+0.0072, +0.0100]} & {\tiny [-8.219, -3.446]} \\[-0.15em]
\addlinespace[0.35em]
GLM with $S$ & -0.0080 & +0.0016 & +0.0096 & -12.331 \\
 & {\tiny [-0.0091, -0.0068]} & {\tiny [+0.0000, +0.0031]} & {\tiny [+0.0075, +0.0113]} & {\tiny [-11.640, -7.694]} \\[-0.15em]
\addlinespace[0.35em]
XGBoost without $S$ & -0.0044 & +0.0000 & +0.0044 & -6.025 \\
 & {\tiny [-0.0052, -0.0038]} & {\tiny [-0.0007, +0.0006]} & {\tiny [+0.0040, +0.0050]} & {\tiny [-6.601, -1.527]} \\[-0.15em]
\addlinespace[0.35em]
XGBoost with $S$ & -0.0141 & -0.0031 & +0.0110 & -9.830 \\
 & {\tiny [-0.0151, -0.0132]} & {\tiny [-0.0042, -0.0021]} & {\tiny [+0.0102, +0.0120]} & {\tiny [-10.044, -5.372]} \\[-0.15em]
\addlinespace[0.35em]
Reference benchmark & -0.1167 & -0.0087 & +0.1079 & -11.974 \\
 & {\tiny [-0.1178, -0.1156]} & {\tiny [-0.0107, -0.0084]} & {\tiny [+0.1059, +0.1085]} & {\tiny [-13.185, -5.128]} \\[-0.15em]
\addlinespace[0.35em]
Benchmark shrinkage & -0.0116 & -0.0034 & +0.0082 & -7.185 \\
 & {\tiny [-0.0124, -0.0107]} & {\tiny [-0.0043, -0.0026]} & {\tiny [+0.0079, +0.0084]} & {\tiny [-7.912, -3.078]} \\[-0.15em]
\addlinespace[0.35em]
\bottomrule
\end{tabular}

\begin{minipage}{0.96\textwidth}
\tiny
\emph{Notes:} The bootstrap is conditional on the fitted benchmark, 
the original tariff vectors, and their fair post-processed versions. 
The same exponential multiplier weights are used for the original and 
fair versions within each replication, so Panel B reports paired 
bootstrap differences.
\end{minipage}
\end{table}

Panel~A shows that conditional finite-portfolio uncertainty is small relative to
the differences between the pricing rules. The coupled coordinates of the two
GLM tariffs are about $0.32$, whereas those of the XGBoost tariffs are about
$0.66$ and $0.68$, with narrow intervals. The opposite ordering is observed for
the marginal Wasserstein coordinate: the GLM tariffs have values close to $0.89$,
compared with approximately $0.81$ for XGBoost. Consequently, the normalized
allocation mismatch is about $0.57$ for the GLMs but only $0.13$--$0.15$ for
XGBoost. The distinction between marginal differentiation and individual
allocation is therefore not an artifact of finite-portfolio variation.

The inclusion of the sensitive attribute has only a modest effect on the reported
coordinates. For the GLM tariffs, the intervals with and without $S$ overlap
substantially. For XGBoost, including $S$ produces a modest increase in coupled
profiling, from $0.655$ to $0.680$, while leaving the marginal coordinate close to
$0.81$. This is consistent with the variable-level diagnostics: the sensitive
attribute carries limited incremental profiling information once the remaining
covariates are available.

Panel~B confirms that the group-parity post-processing operates primarily through
premium allocation rather than through a large change in marginal dispersion.
The paired interval for $\Delta D_1^{\mathrm{cpl}}$ is strictly negative for every
pricing rule. The effect is largest at the reference benchmark, where coupled
profiling falls materially while the marginal Wasserstein coordinate changes only
slightly. For the fitted tariffs, the losses of coupled profiling are smaller but
remain clearly different from zero, and the mismatch increases in every case.
Changes in $D_1^W$ are comparatively small and may be slightly positive, negligible,
or negative depending on the tariff.

The group Wasserstein gaps also decline in every case. Their bootstrap intervals
are wider and more asymmetric than those of the pooling--profiling coordinates.
This reflects the conditional design of the exercise: the parity-corrected tariff vectors are
estimated once on the original portfolio and then evaluated under perturbed
portfolio weights. The results should therefore be read as uncertainty in the
performance of fixed tariff vectors under finite-portfolio composition changes,
rather than as full sampling uncertainty for the estimation and post-processing
procedures.

This conditional bootstrap does not incorporate benchmark estimation,
hyperparameter selection, or the choice of cross-fitting folds. Those sources of
uncertainty are conceptually distinct. Sensitivity to the profiling frontier is
examined through the benchmark ladder in
the benchmark-sensitivity analysis in the main text; a full refit bootstrap would require
re-estimating all candidate tariffs and benchmarks in every replication.

\section{Proofs of the main finite-portfolio results}
\label{sup:main-proofs}

\subsection{Normalization and canonical-path calibration}

For the coupled coordinate,
\[
\bpi^{(a)}-\btheta=(1-a)(\bar\theta\boldsymbol 1-\btheta),
\]
so positive homogeneity of the weighted $L^p$ norm gives $D_p^{\rm cpl}(\bpi^{(a)})=a$. The endpoint characterization, scale invariance, relabelling invariance and continuity are immediate away from a zero denominator.

For the marginal coordinate, the quantile function on the canonical path is
\[
Q_{\bpi^{(a)}}(u)=(1-a)\bar\theta+aQ_{\btheta}(u).
\]
Hence
\[
W_p^p(\mu_{\bpi^{(a)}},\mu_{\btheta})
=(1-a)^p\int_0^1|\bar\theta-Q_{\btheta}(u)|^p\,du,
\]
which yields $D_p^W(\bpi^{(a)})=a$. The remaining properties follow from the metric properties of $W_p$ and from the fact that relabelling leaves the empirical measures unchanged.

\subsection{Transfer interpretation}

If $\bpi$ is balanced, then $\sum_iw_i(\pi_i-\theta_i)=0$. Therefore the weighted positive and negative parts of $\pi_i-\theta_i$ have the same total mass, and
\[
T(\bpi;\btheta,\bw)=\frac12\sum_iw_i|\pi_i-\theta_i|.
\]
The same identity holds at the pooling endpoint. Substitution into the definition of $D_1^{\rm cpl}$ gives the stated transfer formula.

\subsection{Ordering and allocation mismatch}

The discrete matching that pairs each premium $\pi_i$ with its own benchmark value $\theta_i$ is an admissible coupling between $\mu_{\bpi}$ and $\mu_{\btheta}$. Hence
\[
W_p^p(\mu_{\bpi},\mu_{\btheta})
\leq
\sum_iw_i|\pi_i-\theta_i|^p
=
\|\bpi-\btheta\|_{p,\bw}^p.
\]
Because the two coordinates have the same denominator, $D_p^W\geq D_p^{\rm cpl}$. Dividing the difference of the $p$th powers by the common denominator gives
\[
M_p=(1-D_p^{\rm cpl})^p-(1-D_p^W)^p,
\]
and $M_1=D_1^W-D_1^{\rm cpl}$ follows immediately.

\clearpage
\bibliographystyle{plainnat}
\bibliography{biblio}

\end{document}